\documentclass{article}

\usepackage{color}
\usepackage{amsmath,amsthm}
\usepackage{amsfonts,amssymb}
\usepackage[draft]{fixme}
\usepackage{tikz}
\usetikzlibrary{backgrounds,circuits,circuits.ee.IEC,shapes,fit,matrix}

\usetikzlibrary{arrows,positioning}
\tikzstyle{main node} =[circle,fill=white!20,draw,font=\sffamily\Large\bfseries]
\tikzstyle{terminal}=[circle,fill=white!20,draw,font=\sffamily\Large\bfseries,color=purple,fill=none]

\usepackage{comment}
\usepackage[all,2cell,cmtip]{xy}
\UseTwocells

\definecolor{myurlcolor}{rgb}{0.6,0,0}
\definecolor{mycitecolor}{rgb}{0,0,0.8}
\definecolor{myrefcolor}{rgb}{0,0,0.8}
\usepackage[pagebackref]{hyperref}
\hypersetup{colorlinks,
linkcolor=myrefcolor,
citecolor=mycitecolor,
urlcolor=myurlcolor}

\newcommand{\CC}{\mathtt{C}}
\newcommand{\D}{\mathtt{D}}
\newcommand{\FinSet}{\mathtt{FinSet}}
\newcommand{\FinPopSet}{\mathtt{FinPopSet}}
\newcommand{\Set}{\mathtt{Set}}
\newcommand{\Mark}{\mathtt{Mark}}

\newcommand{\Circ}{\mathtt{Circ}}

\newcommand{\LinRel}{ \mathtt{LinRel}}
\newcommand{\LagRel}{ \mathtt{LagRel}}

\newcommand{\DetBalMark}{\mathtt{DetBalMark}}
\newcommand{\Cospan}{\mathtt{Cospan}}

\newcommand{\define}[1]{{\bf \boldmath{#1}}}
\newcommand{\maps}{\colon}

\newcommand{\R}{\mathbb{R}}

\newcommand{\Graph}{\mathrm{Graph}}

\newcommand{\beq}{\begin{equation}}
\newcommand{\eeq}{\end{equation}}

\theoremstyle{plain}

\newtheorem{thm}{Theorem}

\newtheorem{lem}[thm]{Lemma}
\newtheorem{prop}[thm]{Proposition}
\newtheorem{cor}[thm]{Corollary}

\newtheorem{defn}[thm]{Definition}

\theoremstyle{remark}

\begin{document}



\begin{center}   
  {\bf A Compositional Framework for Markov Processes \\}   
  \vspace{0.3cm}
  {\em John\ C.\ Baez \\}
  \vspace{0.3cm}
  {\small
 Department of Mathematics \\
    University of California \\
  Riverside CA, USA 92521 \\ and \\
 Centre for Quantum Technologies  \\
    National University of Singapore \\
    Singapore 117543  \\    } 
  \vspace{0.4cm}
  {\em Brendan Fong \\}
  \vspace{0.3cm}
  {\small Department of Computer Science  \\
    University of Oxford  \\
  United Kingdom OX1 3QD  \\ }
 \vspace{0.3cm}
{\em Blake S. Pollard \\ }
\vspace{0.3cm}
{\small Department of Physics and Astronomy \\
University of California \\
Riverside CA 92521 \\ }
  \vspace{0.3cm}   
  {\small email:  baez@math.ucr.edu, brendan.fong@cs.ox.ac.uk, bpoll002@ucr.edu\\} 
  \vspace{0.3cm}   
  {\small February 15, 2016}
  \vspace{0.3cm}   
\end{center}   

\begin{abstract}
\vskip 0.2em \noindent 
We define the concept of an `open' Markov process, or more precisely, continuous-time Markov chain, which is one where probability can flow in or out of certain states called `inputs' and `outputs'.  One can build up a Markov process from smaller open pieces.  This process is formalized by making open Markov processes into the morphisms of a dagger compact category.  We show that the behavior of a detailed balanced open Markov process is determined by a principle of minimum dissipation, closely related to Prigogine's principle of minimum entropy production. Using this fact, we set up a functor mapping open detailed balanced Markov processes to open circuits made of linear resistors.  We also describe how to `black box' an open Markov process, obtaining the linear relation between input and output data that holds in any steady state, including nonequilibrium steady states with a nonzero flow of probability through the system.  We prove that black boxing gives a symmetric monoidal dagger functor sending open detailed balanced Markov processes to Lagrangian relations between symplectic vector spaces.   This allows us to compute the steady state behavior of an open detailed balanced Markov process from the behaviors of smaller pieces from which it is built.   We relate this black box functor to a previously constructed black box functor for circuits.
\end{abstract}

\section{Introduction}\label{sec:intro}

A continuous-time Markov chain is a way to specify the dynamics of a population which is spread across some finite set of states. Population can flow between the states. The larger the population of a state, the more rapidly population flows out of the state.  Because of this property, under certain conditions the populations of the states tend toward an equilibrium where at any state the inflow of population is balanced by its outflow.  In applications to statistical mechanics, we are often interested in equilibria such that for any two states connected by an edge, say $i$ and $j$, the flow from $i$ to $j$ equals the flow from $j$ to $i$.  A continuous-time Markov chain with a chosen equilibrium of this form is called `detailed balanced'.  

In an electrical circuit made of linear resistors, charge can flow along wires.  In equilibrium, without any driving voltage from outside, the current along each wire will be zero, and the potential at each node will be equal.  This sets up an analogy between detailed balanced continuous-time Markov chains and electrical circuits made of linear resistors:

\vskip 1em
\begin{center}
\begin{tabular}[h]{|c|c|} \hline
\ \bf Circuits & \bf Detailed balanced Markov processes \\ \hline
Potential & Population \\ \hline
Current   & Flow \\ \hline
Conductance & Rate constant \\ \hline
Power & Dissipation \\ \hline
\end{tabular}
\end{center}

\vskip 1em 
\noindent
In this chart and all that follows, we call continuous-time Markov chains `Markov processes', just as an abbreviation.  They are really just an example of the more general concept of Markov process.  

This analogy is already well known \cite{Kelly,Kingman,Nash-Williams}, and Schnakenberg used it in network analysis of the master equation and biological systems \cite{SchnakenRev,SchnakenBook}, so our goal is to formalize and exploit it.  This analogy extends from systems in equilibrium to the more interesting case of nonequilibrium steady states, which are the topic of this paper.     

We have recently introduced a framework for `black boxing' a circuit and extracting the relation it determines between potential-current pairs at the input and output terminals \cite{BaezFong}.  This relation describes the circuit's external behavior as seen by an observer who can only perform measurements at the terminals.  An important fact is that black boxing is `compositional': if one builds a circuit from smaller pieces, the external behavior of the whole circuit can be determined from the external behaviors of the pieces.

Here we adapt this framework to detailed balanced Markov processes.  To do this
we consider `open' Markov processes.  In these, the total population is not
conserved: instead, population is allowed to flow in or out of certain
designated input and output states, or `terminals'.  We explain how to black
box any detailed balanced Markov process, obtaining a relation between
population--flow pairs at the terminals. By the `flow' at a state, we more
precisely mean the net population outflow. This relation holds not only in
equilibrium, but also in any nonequilibrium steady state.  Thus, black boxing an
open detailed balanced Markov process gives its steady state dynamics as seen by
an observer who can only measure populations and flows at the terminals.

At least since the work of Prigogine \cite{GP,Prigogine}, it is widely accepted that a large class of systems minimize entropy production in a nonequilibrium steady state.  However, the precise boundary of this class of systems, and even the meaning of this `principle of minimum entropy production', is much debated \cite{BMN,Landauer,Landauer2,MN}.  Lebon and Jou \cite{LJ} give an argument for it based on four conditions:
\begin{itemize}
\item 
 time-independent boundary conditions,
\item
linear phenomenological laws,
\item
constant phenomenological coefficients,
\item
symmetry of the phenomenological coefficients. 
\end{itemize}
The systems we consider are rather different from those often considered in nonequilibrium thermodynamics, where the `phenomenological laws' are often approximate descriptions of a macroscopic system that has a complicated microscopic structure.  However, the four conditions do apply to our systems. We only consider time-independent boundary conditions.  In our situation the `phenomenological laws' are the relations between flows and populations, and these are indeed linear.  The `phenomenological coefficients' are essentially the rate constants: the constants of proportionality between the flow from state $i$ to state $j$ and the population at state $i$.  These are not symmetric in a naive sense, but the detailed balance condition allows us to express them in terms of a symmetric matrix.  Thus, one should expect a principle of minimum entropy production to hold.  

In fact, we show that a quantity we call the `dissipation' is minimized in any steady state.  This is a quadratic function of the populations and flows, analogous to the power dissipation of a circuit made of resistors.  We make no claim that this quadratic function actually deserves to be called `entropy production'; indeed, Schnakenberg has convincingly argued that they are only approximately equal \cite{SchnakenRev}.  We plan to clarify this in future work, using the calculations in \cite{Pollard}.

This paper is organized as follows.  Section \ref{sec:overview} is an overview of 
the main ideas.  Section \ref{sec:markov} recalls continuous-time Markov chains, which
for short we call simply `Markov processes'.  Section \ref{sec:open_markov}
defines open Markov processes and the open master equation.  Section
\ref{sec:balance} introduces detailed balance for open Markov
processes.  Section \ref{sec:circuits} recalls the principle of minimum power
for open circuits made of linear resistors, and explains how to black box them.  
Section \ref{sec:dissipation} introduces the principle of minimum dissipation for 
open detailed balanced Markov processes, and describes how to black box these.  
Section \ref{sec:reduction} states the analogy between circuits and detailed balanced Markov processes in a formal way.   Section
\ref{sec:composing} describes how to compose open Markov processes, making them
into the morphisms of a category. Section \ref{sec:composing_detailed} does the
same for detailed balanced Markov processes.  Section \ref{sec:black_boxing}
describes the `black box functor' that sends any open detailed balanced Markov
process to the linear relation describing its external behavior, and recalls the
black box functor for circuits.   Section \ref{sec:reduction_2} makes the
analogy between between open detailed balanced Markov processes and open
circuits even more formal, by making it into a functor.  We prove that together
with the two black box functors, this forms a triangle that commutes up to
natural isomorphism.   In Section \ref{sec:symplectic} we prove that
the linear relations in the image of these black box functors are Lagrangian
relations between symplectic vector spaces, and show that the master equation
can be seen as a gradient flow equation.  We summarize our main findings in Section \ref{sec:conclusions}. A quick tutorial on decorated cospans, a key mathematical device 
in this paper, can be found in Appendix \ref{sec:decorated}.

\section{Overview}
\label{sec:overview}

This diagram summarizes our method of black boxing detailed balanced Markov processes: 
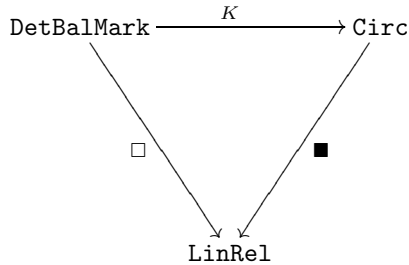
\[
   \xy
   (-20,20)*+{\DetBalMark}="1";
  (20,20)*+{\Circ}="2";
   (0,-10)*+{\LinRel}="5";
        {\ar^{K} "1";"2"};
        {\ar_{\square} "1";"5"};
        {\ar^{\blacksquare} "2";"5"};
\endxy
\]
Here $\DetBalMark$ is the main category of interest.  A morphism in this
category is a detailed balanced Markov process with specified `input' and
`output' states, which serve to define the source and target of the morphism:
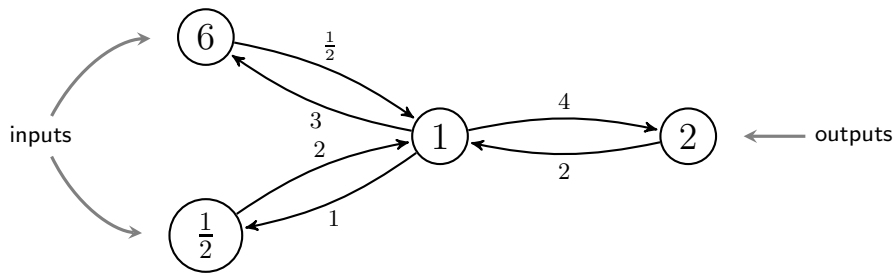
\[
\begin{tikzpicture}[->,>=stealth',shorten >=1pt,thick,scale=1.1]
  \node[main node] (1) at (0,2.2) {$6$};
  \node[main node](2) at (0,-.2) {$\frac{1}{2}$};
  \node[main node](3) at (2.83,1)  {$1$};
  \node[main node](4) at (5.83,1) {$2$};
\node(input) at (-2,1) {\small{\textsf{inputs}}};
\node(output) at (7.83,1) {\small{\textsf{outputs}}};
  \path[every node/.style={font=\sffamily\small}, shorten >=1pt]
    (3) edge [bend left=12] node[above] {$4$} (4)
    (4) edge [bend left=12] node[below] {$2$} (3)
    (2) edge [bend left=12] node[above] {$2$} (3) 
    (3) edge [bend left=12] node[below] {$1$} (2)
    (1) edge [bend left=12] node[above] {$\frac{1}{2}$}(3) 
    (3) edge [bend left=12] node[below] {$3$} (1);
    
\path[color=gray, very thick, shorten >=10pt, ->, >=stealth] (output) edge (4);
\path[color=gray, very thick, shorten >=10pt, ->, >=stealth, bend left] (input) edge (1);
\path[color=gray, very thick, shorten >=10pt, ->, >=stealth, bend right]
(input) edge (2);
\end{tikzpicture}
\]
Call this morphism $M$.  In general each state may be specified as both an input and an output, or as inputs and outputs multiple times. The detailed balanced Markov process itself comprises a finite set of states together with a finite set of edges between them, with each state $i$ labelled by an equilibrium population $q_i >0$, and each edge $e$  labelled by a rate constant $r_e > 0$.  These populations and rate constants are required to obey the detailed balance condition. 

Note that we work with un-normalized probabilities, which we call `populations', rather than probabilities that must sum to 1.  The reason is that probability is not conserved: it can flow in or out at the inputs and outputs.  We allow it to flow both in and out at both the input states and the output states.

Composition in $\DetBalMark$ is achieved by identifying the output states of one
open detailed balanced Markov process with the input states of another. The
populations of identified states must match. For example, we may compose this
morphism $N$:
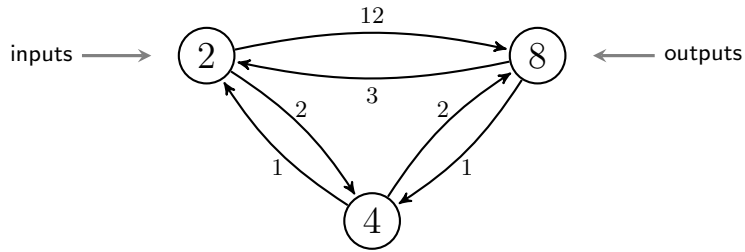
\[
\begin{tikzpicture}[->,>=stealth',shorten >=1pt,thick,scale=1.1]
  \node[main node] (4) at (0,1) {$2$};
  \node[main node](5) at (2,-1) {$4$};
  \node[main node](6) at (4,1) {$8$};
\node(input) at (-2,1) {\small{\textsf{inputs}}};
\node(output) at (6,1) {\small{\textsf{outputs}}};
  \path[every node/.style={font=\sffamily\small}, shorten >=1pt]
    (4) edge [bend left=12] node[above] {$2$} (5)
    (4) edge [bend left=12] node[above] {$12$} (6)
    (5) edge [bend left=12] node[below] {$1$} (4) 
    (5) edge [bend left=12] node[above] {$2$} (6)
    (6) edge [bend left=12] node[below] {$3$}(4) 
    (6) edge [bend left=12] node[below] {$1$} (5);
    
\path[color=gray, very thick, shorten >=10pt, ->, >=stealth] (output) edge (6);
\path[color=gray, very thick, shorten >=10pt, ->, >=stealth] (input) edge (4);
\end{tikzpicture}
\]
with the previously shown morphism $M$ to obtain this morphism $N \circ M$:
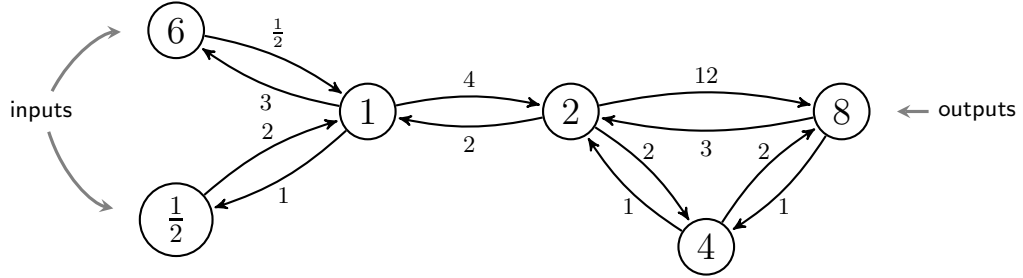
\[
\begin{tikzpicture}[->,>=stealth',shorten >=1pt,thick,scale=.9]
  \node[main node] (1) at (0,2.2) {$6$};
  \node[main node](2) at (0,-.6) {$\frac{1}{2}$};
  \node[main node](3) at (2.83,1)  {$1$};
  \node[main node](4) at (5.83,1) {$2$};
  \node[main node](5) at (7.83,-1) {$4$};
  \node[main node](6) at (9.83,1) {$8$};
\node(input) at (-2,1) {\small{\textsf{inputs}}};
\node(output) at (11.83,1) {\small{\textsf{outputs}}};
  \path[every node/.style={font=\sffamily\small}, shorten >=1pt]
    (3) edge [bend left=12] node[above] {$4$} (4)
    (4) edge [bend left=12] node[below] {$2$} (3)
    (2) edge [bend left=12] node[above] {$2$} (3) 
    (3) edge [bend left=12] node[below] {$1$} (2)
    (1) edge [bend left=12] node[above] {$\frac{1}{2}$}(3) 
    (3) edge [bend left=12] node[below] {$3$} (1)
    (4) edge [bend left=12] node[above] {$2$} (5)
    (4) edge [bend left=12] node[above] {$12$} (6)
    (5) edge [bend left=12] node[below] {$1$} (4) 
    (5) edge [bend left=12] node[above] {$2$} (6)
    (6) edge [bend left=12] node[below] {$3$}(4) 
    (6) edge [bend left=12] node[below] {$1$} (5);
    
\path[color=gray, very thick, shorten >=10pt, ->, >=stealth] (output) edge (6);
\path[color=gray, very thick, shorten >=10pt, ->, >=stealth, bend left] (input) edge (1);
\path[color=gray, very thick, shorten >=10pt, ->, >=stealth, bend right]
(input) edge (2);
\end{tikzpicture}
\]

Our category $\DetBalMark$ is actually a dagger compact category. This makes
other procedures on Markov processes available to us. An important one is
`tensoring', which allows us to model two Markov processes in parallel. This 
lets us take $M$ and $N$ above and `set them side by side', giving $M
\otimes N$:
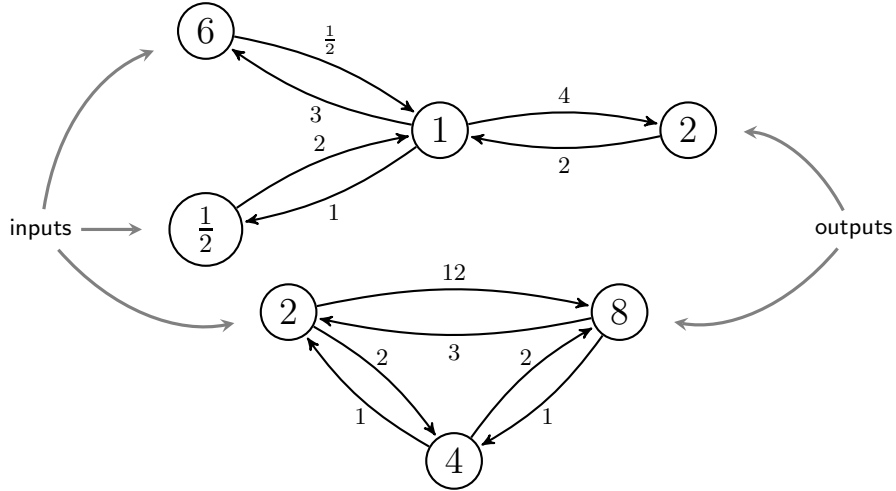
\[
\begin{tikzpicture}[->,>=stealth',shorten >=1pt,thick,scale=1.1]
  \node[main node] (1) at (0,2.2) {$6$};
  \node[main node](2) at (0,-.2) {$\frac{1}{2}$};
  \node[main node](3) at (2.83,1)  {$1$};
  \node[main node](4) at (5.83,1) {$2$};
  \node[main node](4a) at (1,-1.2) {$2$};
  \node[main node](5) at (3,-3) {$4$};
  \node[main node](6) at (5,-1.2) {$8$};
\node(input) at (-2,-.2) {\small{\textsf{inputs}}};
\node(output) at (7.83,-.2) {\small{\textsf{outputs}}};
  \path[every node/.style={font=\sffamily\small}, shorten >=1pt]
    (3) edge [bend left=12] node[above] {$4$} (4)
    (4) edge [bend left=12] node[below] {$2$} (3)
    (2) edge [bend left=12] node[above] {$2$} (3) 
    (3) edge [bend left=12] node[below] {$1$} (2)
    (1) edge [bend left=12] node[above] {$\frac{1}{2}$}(3) 
    (3) edge [bend left=12] node[below] {$3$} (1)
    (4a) edge [bend left=12] node[above] {$2$} (5)
    (4a) edge [bend left=12] node[above] {$12$} (6)
    (5) edge [bend left=12] node[below] {$1$} (4a) 
    (5) edge [bend left=12] node[above] {$2$} (6)
    (6) edge [bend left=12] node[below] {$3$}(4a) 
    (6) edge [bend left=12] node[below] {$1$} (5);
    
\path[color=gray, very thick, shorten >=10pt, ->, >=stealth, bend right] (output) edge (4);
\path[color=gray, very thick, shorten >=10pt, ->, >=stealth, bend left] (input) edge (1);
\path[color=gray, very thick, shorten >=10pt, ->, >=stealth] (input) edge (2);
\path[color=gray, very thick, shorten >=10pt, ->, >=stealth, bend left] (output) edge (6);
\path[color=gray, very thick, shorten >=10pt, ->, >=stealth, bend right] (input) edge (4a);
\end{tikzpicture}
\]
Another important structure is the so-called compactness. This allows us to
take some inputs of an open Markov process and consider them instead as outputs, or
vice versa. For example, using the compactness of $\DetBalMark$ we may obtain
this open Markov process from $M$:
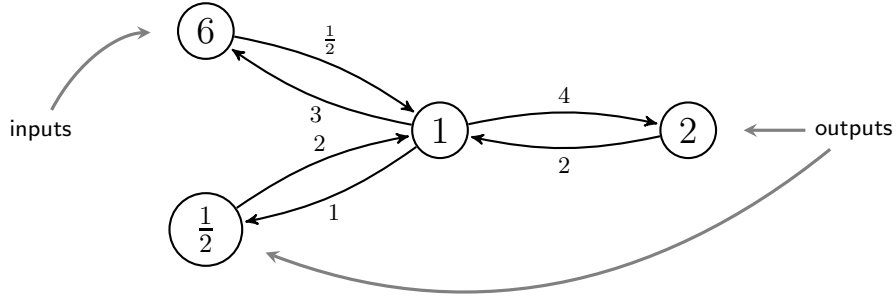
\[
\begin{tikzpicture}[->,>=stealth',shorten >=1pt,thick,scale=1.1]
  \node[main node] (1) at (0,2.2) {$6$};
  \node[main node](2) at (0,-.2) {$\frac{1}{2}$};
  \node[main node](3) at (2.83,1)  {$1$};
  \node[main node](4) at (5.83,1) {$2$};
\node(input) at (-2,1) {\small{\textsf{inputs}}};
\node(output) at (7.83,1) {\small{\textsf{outputs}}};
  \path[every node/.style={font=\sffamily\small}, shorten >=1pt]
    (3) edge [bend left=12] node[above] {$4$} (4)
    (4) edge [bend left=12] node[below] {$2$} (3)
    (2) edge [bend left=12] node[above] {$2$} (3) 
    (3) edge [bend left=12] node[below] {$1$} (2)
    (1) edge [bend left=12] node[above] {$\frac{1}{2}$}(3) 
    (3) edge [bend left=12] node[below] {$3$} (1);
    
\path[color=gray, very thick, shorten >=10pt, ->, >=stealth] (output) edge (4);
\path[color=gray, very thick, shorten >=10pt, ->, >=stealth, bend left] (input) edge (1);
\path[color=gray, very thick, shorten >=10pt, ->, >=stealth, bend left]
(output) edge (2);
\end{tikzpicture}
\]
In fact all our other categories are dagger compact categories too,
and our functors preserve this structure.  Dagger compact categories are a well-known
framework for describing systems with inputs and outputs \cite{AC,BaezStay,Se}. 

A morphism in the category $\Circ$ is an electrical circuit made of resistors: that is,
a (directed) graph with each edge labelled by a `conductance' $c_e > 0$, again with specified input and output nodes:
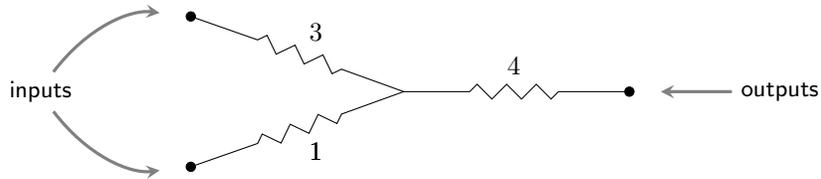
\[
\begin{tikzpicture}[circuit ee IEC, set resistor graphic=var resistor IEC graphic]
\node[contact] (I1) at (0,2) {};
\node[contact] (I2) at (0,0) {};
\node[contact] (O1) at (5.83,1) {};
\node(input) at (-2,1) {\small{\textsf{inputs}}};
\node(output) at (7.83,1) {\small{\textsf{outputs}}};
\draw (I1) 	to [resistor] node [label={[label distance=2pt]85:{$3$}}] {} (2.83,1);
\draw (I2)	to [resistor] node [label={[label distance=2pt]275:{$1$}}] {} (2.83,1)
				to [resistor] node [label={[label distance=3pt]90:{$4$}}] {} (O1);
\path[color=gray, very thick, shorten >=10pt, ->, >=stealth, bend left] (input) edge (I1);		\path[color=gray, very thick, shorten >=10pt, ->, >=stealth, bend right] (input) edge (I2);		
\path[color=gray, very thick, shorten >=10pt, ->, >=stealth] (output) edge (O1);
\end{tikzpicture}
\]

Finally, a morphism in the category $\LinRel$ is a linear relation $F \maps U \leadsto V$ between finite-dimensional real vector spaces $U$ and $V$; this is nothing but a linear subspace of $U \oplus V$.  In earlier work \cite{BaezFong} we introduced the category $\Circ$ and the `black box functor' 
\[    \blacksquare \maps \Circ \to \LinRel . \]
The idea is that any circuit determines a linear relation between the potentials and net current flows at the inputs and outputs.   This relation describes the behavior of a circuit of resistors as seen from outside.  

The functor $K$ converts a detailed balanced Markov process into an electrical circuit made of resistors.  This circuit is carefully chosen to reflect the steady-state behavior of the Markov process.   Its underlying graph is the same as that of the Markov process, so the `states' of the Markov process are the same as the `nodes' of the circuit.
Both the equilibrium populations at states of the Markov process and the rate constants labelling edges of the Markov process are used to compute the conductances of edges
of this circuit.  In the simple case where the Markov process has exactly one edge from
any state $i$ to any state $j$, the rule is
\[                 C_{i j} = H_{i j} q_j   \]
where:
\begin{itemize}
\item $q_j$ is the equilibrium population of the $j$th state of the Markov process,
\item  $H_{i j}$ is the rate constant for the edge from the $j$th state to the $i$th state of the Markov process, and
\item  $C_{i j}$ is the conductance (that is, the reciprocal of the resistance) of the wire from the $j$th node to the $i$th node
of the resulting circuit.
\end{itemize}
The detailed balance condition for Markov processes says precisely that
the matrix $C_{i j}$ is symmetric.  This is just right for an electrical circuit made of resistors, since it means that the resistance of the wire from node $i$ to node $j$ 
equals the resistance of the same wire in the reverse direction, from node $j$ to node $i$.

The functor $\square$ is the main new result of this paper.  It maps any detailed balanced Markov process to the linear relation obeyed by populations and flows at the inputs and outputs in a steady state.  In short, it describes the steady state behavior of the Markov process `as seen from outside'.  We draw this functor as a white box merely to
distinguish it from the other black box functor.

The triangle of functors thus constructed does not commute!   However, 
a general lesson of category theory is that we should only expect diagrams of
functors to commute \emph{up to natural isomorphism}, and this is what happens here:
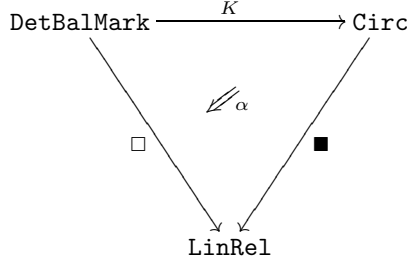
\[
   \xy
   (-20,20)*+{\DetBalMark}="1";
  (20,20)*+{\Circ}="2";
   (0,-10)*+{\LinRel}="5";
        {\ar^{K} "1";"2"};
        {\ar_{\square} "1";"5"};
        {\ar^{\blacksquare} "2";"5"};
        {\ar@{=>}^<<{\scriptstyle \alpha} (1,11); (-3,8)};
\endxy
\]
This `corrects' the black box functor for resistors to give the one for detailed
balanced Markov processes.   The functors $\square$ and $\blacksquare
\circ K$ are equal on objects.  An object in $\DetBalMark$ is a finite set $X$
with each element $i \in X$ labelled by a positive population $q_i$; both these
functors map such an object to the vector space $\R^X \oplus \R^X$.    For the
functor $\square$, we think of this as a space of population-flow pairs.  For
the functor $\blacksquare \circ K$, we think of it as a space of
potential-current pairs, since $K$ converts Markov processes to circuits made of
resistors.  The natural transformation $\alpha$ then gives a linear relation 
\[ \alpha_{X,q} \maps \R^X \oplus \R^X \leadsto \R^X \oplus \R^X ,\]
in fact an isomorphism of vector spaces, which converts potential-current pairs
into population-flow pairs in a manner that depends on the $q_i$.  This
isomorphism maps any $n$-tuple of potentials and currents $(\phi_i, \iota_i)$
into the $n$-tuple of populations and flows $(p_i, j_i)$ given by
\[          p_i = \phi_i q_i,  \qquad   j_i = \iota_i  .\]

The naturality of $\alpha$ actually allows us to reduce the problem of computing
the functor $\square$ to the problem of computing $\blacksquare$.   Suppose $M
\maps (X,q) \to (Y,r)$ is any morphism in $\DetBalMark$.  The object $(X,q)$ is
some finite set $X$ labelled by populations $q$, and $(Y,r)$ is some finite set
$Y$ labelled by populations $r$.  Then the naturality of $\alpha$ means that
this square commutes:
\[
  \xymatrix{
    \R^X \oplus \R^X \ar[rr]^{\blacksquare K(M)} \ar[dd]_{\alpha_{X,q}} &&  \R^Y \oplus \R^Y 
    \ar[dd]^{\alpha_{Y,r}} \\ \\ 
    \R^X \oplus \R^X \ar[rr]^{\square(M)} &&  \R^Y \oplus \R^Y 
  }
\]
Since $\alpha_{X,q}$ and $\alpha_{Y,r}$ are isomorphisms, we can solve for the
functor $\square$:
\[   \square(M) = \alpha_{Y,r}\circ \blacksquare K(M) \circ \alpha_{X,q}^{-1}  .
\]
This equation has a clear intuitive meaning: it says that to compute the behavior of
a detailed balanced Markov process, namely $\square(M)$, we convert it into a circuit made of resistors and compute the behavior of that, namely $\blacksquare K(M)$.  This is
not \emph{equal} to the behavior of the Markov process, but we can compute that 
behavior by converting the input populations and flows into potentials and
currents, feeding them into our circuit, and then converting the outputs back into populations and flows. 

\section{Markov processes}       
\label{sec:markov}

We define a `Markov process' to be a graph with nodes labelled by `populations'
and edges labelled by `rate constants'.   More precisely:

\begin{defn}
A \define{Markov process} $M$ is a diagram
\[ \xymatrix{ (0,\infty) & E \ar[l]_-r \ar[r]<-.5ex>_t  \ar[r] <.5ex>^s & N  }  \]
where $N$ is a finite set of \define{nodes} or \define{states}, $E$ is a finite set of \define{edges}, $s,t \maps E \to N$ assign to each edge its \define{source} and \define{target}, and $r \maps E \to (0,\infty)$ assigns a \define{rate constant} $r_e$ to each edge $e \in E$. 

In this situation we call $M$ a \define{Markov process on} $N$.  If $e \in E$ has source $i$ and target $j$, we write $e \maps i \to j$.
\end{defn}

From a Markov process on $N$ we can construct a square matrix of real numbers, or more precisely a function $H \maps N \times N \to \R$, called its \define{Hamiltonian}.
If $i \ne j$ we define
\[ H_{i j} = \sum_{e \maps j \to i} r_e  \]
to be the sum of the rate constants of all edges from $j$ to $i$.   We choose the diagonal
entries in a more subtle way:
\[  H_{i i} = -\sum_{\substack{e \maps i \to j \\ j \ne i}} r_e . \]
Given a probability distribution $p$ on $N$, for $i \ne j$ we interpret $H_{ij}p_j $ as the rate at which population flows from state $j$ to state $i$, while the quantity $H_{ii}p_i$ is the outflow of population from state $i$.  We choose the diagonal entries $H_{ii}$ in a way that ensures total population is conserved.   This means that $H$ should be `infinitesimal stochastic'.

\begin{defn} An $n \times n$ matrix $T$ is \define{stochastic} if it maps probability 
distributions to probability distributions: i.e., if the vector $p \in \R^n$ has entries $p_i \ge 0$ and
$\sum_i p_i = 1$, then the same holds for the vector $T p$.
\end{defn}

\begin{defn} An $n \times n$ matrix $A$ is \define{infinitesimal stochastic} if
$\exp(t A)$ is stochastic for all $t \ge 0$.
\end{defn}

\begin{lem} \label{lem.infstoch}
  An $n \times n$ matrix $A$ is infinitesimal stochastic if and only if its off-diagonal entries are non-negative and the entries in each column sum to zero:
\[ A_{ij} \geq 0 \hspace{1em} \textrm{ if } i \neq j \ \\ \]
\[ \sum_i A_{ij} = 0.\] 
\end{lem}

\begin{proof} This is well known \cite{BB}.
\end{proof}

\begin{lem}  If $M$ is a Markov process and $H$ is its Hamiltonian, then $H$ is
infinitesimal stochastic.
\end{lem}

\begin{proof} This is well known \cite{BB}, but easy to show. By definition, the off-diagonal terms of $H$ are nonnegative.  We also have
\[\sum_i H_{ij} = H_{jj}+ \sum_{i \ne j} H_{ij} =- \sum_{\substack{e \maps j \to
i \\ i \ne j}} r_e + \sum_{\substack{e \maps j \to i \\ i \ne j}} r_e  = 0.  \qedhere \]
\end{proof}

\begin{defn} Given a Markov process, the \define{master equation} for a function $p \maps [0,\infty) \to \R^N$ is
\[ \frac{d}{dt}p(t) = Hp(t) \]
where $H$ is the Hamiltonian. 
\end{defn}

Solutions of the master equation are of the form
\[    p(t) = \exp(t H) p(0),
\]
where $p(0) \in \R^N$.
Since $\exp(t H)$ is stochastic for $t \ge 0$, if $p(t)$ is a probability distribution 
at time $t$ then it will be a probability distribution at all later times. This is not true for the `open' master equation discussed in the next section, and that is why we need to work with more general `populations' that do not obey $\sum_i p_i =1$.

\section{Open Markov processes}
\label{sec:open_markov}

We will be interested in `open' Markov processes, where population can flow in
or out of certain nodes.  One reason is that non-equilibrium steady states of
open Markov processes play a fundamental role in stochastic thermodynamics.  In
Section \ref{sec:composing}, we will show that open Markov processes can be seen as
morphisms of a category.  Composing these morphisms gives a way to build Markov
processes out of smaller open pieces.  

To create a category where morphisms are, roughly, open Markov processes, we use the formalism of decorated cospan categories \cite{Fong}. Decorated cospans provide a powerful way to describe the interfacing of systems with inputs and outputs. A cospan describes the inputs and outputs of a system. The `decoration' carries the complete description of this system.

Mathematically, a \define{cospan} of sets consists of a set $N$ together with functions
$i \maps X \to N$ and $o \maps Y \to N$.  We draw a cospan as follows:
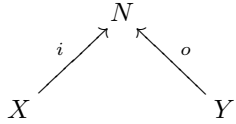
\[ \xymatrix{ & N & \\ 
X \ar[ur]^{i} & &  Y \ar[ul]_{o} \\ } \]
The set $N$ describes the system, $X$ describes its inputs, and $Y$ its outputs.  
We shall be working with a number of examples, but for open Markov processes $N$ is the set of states of some Markov process, and the maps $i \maps X \to N$ and $o \maps Y \to N$ specify how the input and output states are included in $N$.   We do not require these maps to be one-to-one.

We then `decorate' the cospan with a complete description of the system. For example, to decorate the above cospan with a Markov process, we attach to it the extra data 
\[ \xymatrix{ (0,\infty) & E \ar[l]_-r \ar[r]<-.5ex>_t \ar[r] <.5ex>^s & N  } \] 
describing a Markov process with $N$ as its set of states.  Thus, we make the following definition:

\begin{defn}
Given finite sets $X$ and $Y$, an \define{open Markov process from} $X$ \define{to} $Y$ is a cospan of finite sets
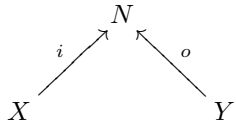
\[ \xymatrix{ & N & \\
X \ar[ur]^{i} && Y \ar[ul]_{o} \\ } \]
together with a Markov process $M$ on $N$.  We often abbreviate such an open Markov process simply as $M \maps X \to Y$.  We say $X$ is the set of \define{inputs} of the open Markov process and $Y$ is its set of \define{outputs}. We define a \define{terminal} to 
be a node in $T = i(X) \cup o(Y)$, and call a node 
\define{internal} if it is not a terminal.
\end{defn}

The key new feature of an open Markov process is that population can flow in and out of its terminals. To describe these modified dynamics we introduce the `open master equation'.   Here the populations at the terminals are specified functions of time, while the populations at internal nodes obey the usual master equation:

\begin{defn}
Given an open Markov process $M \maps X \to Y$ consisting of a cospan of finite sets
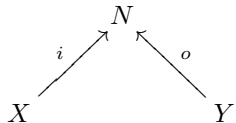
\[ \xymatrix{ & N & \\
X \ar[ur]^{i} && Y \ar[ul]_{o} \\ } \]
together with a Markov process
\[ \xymatrix{ (0,\infty) & E \ar[l]_-r \ar[r]<-.5ex>_t \ar[r] <.5ex>^s & N }  \]
on $N$, we say
a time-dependent population $p(t) \maps N \to (0,\infty)$ (where $t \in
[0,\infty)$) 
is a solution of the \define{open master equation} with boundary conditions 
$f(t) \maps T \to [0,\infty)$ if 
\[ \begin{array}{ccll}\displaystyle{ \frac{d}{dt}p_i(t) } &=& \displaystyle{
\sum_j H_{ij} p_j(t)}, &  i \in N-T \\  \\
 p_i(t) &=& f_i(t), & i \in T. 
\end{array}\]
\end{defn}

We will be especially interested in `steady state' solutions of the open master equation:

\begin{defn}
A \define{steady state} solution of the open master equation is a solution $p(t) \maps N \to [0,\infty)$ such that $\frac{dp}{dt} = 0$. 
\end{defn}

In Section \ref{sec:balance} we turn to open Markov processes obeying the principle of detailed balance.  In Section \ref{sec:dissipation} we show that for these, steady state solutions of the open master equation minimize a certain quadratic form.  This is analogous to the minimization of power dissipation by circuits made of linear resistors.   In Section \ref{sec:reduction_2}, this analogy lets us reduce the black boxing problem for detailed balanced Markov processes to the analogous, and already solved, problem for circuits of resistors.

\section{Detailed balance} 
\label{sec:balance}

We are especially interested in Markov processes where the detailed balance 
condition holds for a specific population $q$.  This means that with this particular choice
of $q$, the flow of population from $i$ to $j$ equals the flow from $j$ to $i$.  This ensures that $q$ is an equilibrium, but it is a significantly stronger condition.

\begin{defn}
A \define{Markov process with populations} 
\[ \xymatrix{  (0,\infty) & E \ar[l]_-r \ar[r]<-.5ex>_t  \ar[r] <.5ex>^s & N 
\ar[r]^-q & (0,\infty)  }  \]
is a Markov process
\[ \xymatrix{   (0,\infty) & E \ar[l]_-r \ar[r]<-.5ex>_t  \ar[r] <.5ex>^s & N 
  }  \]
together with a map $q \maps N \to (0,\infty)$ assigning to each state a \define{population}. 
\end{defn}

\begin{defn}
We say a Markov process with populations is in \define{equilibrium} if $H q = 0$
where $H$ is the Hamiltonian associated to that Markov process.
\end{defn}

\begin{defn}
We say a Markov process with populations is a \define{detailed balanced Markov process} if 
\[     H_{ij} q_j = H_{ji} q_i 
\]
for all nodes $i, j \in N$. We call the above equation the \define{detailed balance condition}.
\end{defn}

In terms of the edges of the Markov process, we may rewrite the detailed balance
condition as
\[      
\sum_{e \maps i \to j} r_e q_i = \sum_{e \maps j \to i} r_e q_j 
\] 
for all $i,j \in N$.

\begin{prop} 
  If a Markov process with populations obeys the detailed balance condition, it
  is in equilibrium.  
\end{prop}

\begin{proof}
  This is well known \cite{BB}. For the sake of completeness, however, recall
  that $H$ is infinitesimal stochastic, so by Lemma \ref{lem.infstoch} we have
  $H_{ii} = -\sum_{j \ne i} H_{ji}$. Then given the population $q$, at state $i$
  the master equation reads
  \[ 
    \frac{dq_i}{dt} = \sum_{j} H_{ij} q_j =  H_{ii}q_i + \sum_{j \neq i} H_{ij}q_j
    = \sum_{j \neq i} \Big(H_{ij}q_j-H_{ji}q_i\Big).  
  \]
  By detailed balance, this is zero.
\end{proof}

We are particularly interested in `open' detailed balanced Markov processes. These go between finite sets with populations. 

\begin{defn}
A \define{finite set with populations} is a finite set $X$ together with a map
$q\maps X \to (0,\infty)$.   We write $q_i \in (0,\infty)$ for the value of $q$
at the point $i \in X$, and call $q_i$ the \define{population} of $i$.
\end{defn}

\begin{defn}
\label{defn:FinPopSet}
Given two finite sets with populations $(X,q)$ and $(X',q')$, we define a \define{map} 
$f \maps (X,q) \to (X',q')$ to be a function $f \maps X \to X'$ preserving the populations: 
$q = q' f$.   We let $\FinPopSet$ be the category of finite sets with populations and maps 
between them.
\end{defn}

Note that any detailed balanced Markov process with populations $M$, say
\[ 
  \xymatrix{  (0,\infty) & E \ar[l]_-r \ar[r]<-.5ex>_t  \ar[r] <.5ex>^s & N 
\ar[r]^-q & (0,\infty) , }  
\]
has an underlying finite set with populations $(N,q)$.  We say $M$ is a detailed balanced Markov process \define{on} $(N,q)$.

\begin{defn}
Given finite sets with populations $(X,u)$ and $(Y,v)$, 
an \define{open detailed balanced Markov process from} $(X,u)$ \define{to} $(Y,v)$ 
is a cospan of finite sets with populations:
\[ \xymatrix{ & (N,q) & \\
(X,u) \ar[ur]^{i} && (Y,v) \ar[ul]_{o} \\ } \]
together with a detailed balanced Markov process $M$ on $(N,q)$.   We often abbreviate this as $M \maps (X,u) \to (Y,v)$.
\end{defn}

Note that as $i$ and $o$ are maps of finite sets with populations, we must have
$u=qi$ and $v=qo$. 

Here is an example:
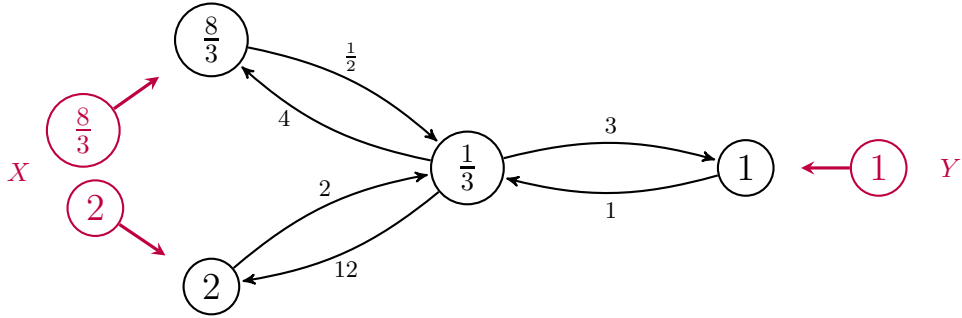
\[
\begin{tikzpicture}[->,>=stealth',shorten >=1pt,auto,node distance=3.7cm,
  thick]
\node[main node](1) {$\frac{8}{3}$};
  \node[main node](2) [below=2.4cm of 1] {$2$};
  \node[main node](3) [below right=1cm and 2.7cm of 1]  {$\frac{1}{3}$};
  \node[main node](4) [right of=3] {$1$};
\node [terminal,below left=0.5cm and 1cm of 1](A) {$\frac{8}{3}$};
\node [terminal,above left=0.5 cm and 1cm of 2](B) {$2$};
\node [above left=1 cm and 2cm of 2,color =purple](X) {$X$};
\node [terminal,right =1 cm of 4](C) {$1$};
\node [right=0.3cm of C,color =purple](Y) {$Y$};
  \path[every node/.style={font=\sffamily\small}, shorten >=1pt]
    (3) edge [bend left=15] node[above] {$3$} (4)
    (4) edge [bend left=15] node[below] {$1$} (3)
    (2) edge [bend left=15] node[above,midway] {$2$} (3) 
    (3) edge [bend left=15] node[below] {$12$} (2)
   (1) edge [bend left=15] node[above] {$\frac{1}{2}$}(3) 
    (3) edge [bend left=15] node[below, near end] {$4$} (1);
    
\path[color=purple, very thick, shorten >=10pt, ->, >=stealth] (C) edge (4);
\path[color=purple, very thick, shorten >=10pt, ->, >=stealth] (A) edge (1);
\path[color=purple, very thick, shorten >=10pt, ->, >=stealth] (B) edge (2);
\end{tikzpicture}
\]

\section{Circuits}
\label{sec:circuits}

There is a long-established analogy between circuits and detailed balanced
Markov processes \cite{Kelly,Kingman,Nash-Williams,SchnakenRev,SchnakenBook}.
Exploiting this analogy lets us reduce the process of black boxing detailed balanced Markov processes to the analogous process for circuits.
First, however, we must introduce some of the fundamentals of circuit theory.  
In what follows we use `circuit' as shorthand for `circuit made of linear
resistors', though our previous work \cite{BaezFong} was more general.

\begin{defn}
A \define{circuit} $C$ is a graph with edges labelled by positive real numbers:
\[ \xymatrix{  (0,\infty) & E \ar[l]_-c \ar[r]<-.5ex>_t  \ar[r] <.5ex>^s & N }  \]
where $N$ is a finite set of \define{nodes}, $E$ a finite set of \define{edges}, $s,t \maps E \to N$ assign to each edge its \define{source} and \define{target}, and $c \maps E \to (0,\infty)$ assigns a \define{conductance} $c_e$ to each edge $e \in E$. 

In this situation we call $C$ a \define{circuit on} $N$.   If $e \in E$ has source $i$ and target $j$, we write $e \maps i \to j$. 
\end{defn}

In our previous work we labelled edges by resistances, which are the reciprocals
of conductances, but translating between circuits and Markov processes will be
easier if we use conductances, since these are more directly related to rate
constants. 

\begin{defn}
Given finite sets $X$ and $Y$, an \define{open circuit from} $X$ \define{to} $Y$ is a cospan of finite sets 
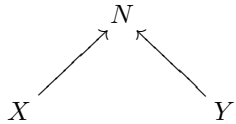
\[ \xymatrix{ & N & \\ X  \ar[ur] & &  Y \ar[ul] } \] 
together with a circuit $C$ on $N$.   We often abbreviate this as $C \maps X \to Y$.
Again we define a \define{terminal} to be a node in $T= i(X) \cup o(Y)$, and call a node in $N-T$ \define{internal}.
\end{defn}

Ulimately, an open circuit is of interest for the relationship it imposes
between the potentials at its inputs and outputs, the net current inflows at its
inputs, and the net current outflows at its outputs. We call this the
\define{behavior} of the open circuit.  To understand a circuit's behavior,
however, it is instructive to consider potentials and currents on all nodes.
These are governed by Ohm's law.

Consider a single edge $e \maps i \to j$ with conductance $c_e \in (0,\infty)$
and with potential $\phi_{s(e)}$ at its source and $\phi_{t(e)}$ at its target.
The voltage $V_e \in \R$ across an edge $e \in E$ is given by 
\[ 
  V_e = \phi_{s(e)}-\phi_{t(e)}. 
\]
If the voltage across an edge is nonzero, then current will flow along that edge:
Ohm's Law states that the current $I_e$ along an edge $e$ is related to the voltage across the edge $V_e$ via
\[ I_e = c_e V_e. \]

\begin{defn}
Given a circuit, we define the \define{extended power functional} $P(\phi) \maps
\R^N \to \R$ to map each potential to one-half the total power dissipated by the
circuit at that potential:
\[ P(\phi) = \frac12 \sum_{e \in E} I_e V_e = \frac{1}{2} \sum_{e \in E} c_e ( \phi_{s(e)} - \phi_{t(e)} )^2. \]
\end{defn}
\noindent The factor of $\frac{1}{2}$ makes certain formulas cleaner.

The extended power functional tells us the power dissipated when the potential at
each node of the circuit is known.  However, when we treat an open circuit as a `black box', we can only access the potentials at terminals and the currents flowing in or out of the terminals.  We call these the \define{boundary potentials} and \define{boundary currents}: they are functions from $T$ to $\R$, where $T$ is 
the set of terminals.

\begin{defn}
We say \define{Kirchhoff's current law} holds at $n \in N$ if 
\[ 
\sum_{e\maps i\to n} I_e \quad = \sum_{e \maps n \to i} I_e. 
\]
\end{defn}
Open circuits are required to obey Kirchhoff's current law at all internal
nodes, but not at the terminals.  To understand the behavior of a circuit, we must then answer the question: what boundary currents may exist given a fixed boundary potential?  

Here the extended power functional comes in handy.  The key property of the extended power functional is that its partial derivative with respect to the potential at a node is equal to the net current outflow at that node:

\begin{prop}
  Let $P\maps \R^N \to \R$ be the extended power functional for some circuit
  $C$. Then for all $n \in N$ we have
  \[
    \frac{\partial P}{\partial \phi_n} = \sum_{e \maps n \to i} I_e - \sum_{e\maps i\to n} I_e. 
  \]
\end{prop}
\begin{proof}
This is a simple computation and application of Ohm's law:
\begin{align*}
\frac{\partial P(\phi)}{\partial \phi_n} &= \sum_{e \maps n \to i} c_e ( \phi_n - \phi_i ) - \sum_{e \maps i \to n} c_e ( \phi_i - \phi_n) \\
&= \sum_{e \maps n \to i} I_e - \sum_{e \maps i \to n} I_e. \qedhere
\end{align*}
\end{proof}

In particular, as open circuits must obey Kirchhoff's current law at internal
nodes, this implies that for a potential on an open circuit to be physically
realizable, the partial derivative of its extended power functional with respect
to the potential at each internal node must vanish. We leverage this fact to
study circuits and detailed balanced Markov processes.

\begin{defn}
We say $\phi \in \R^N$ obeys the \define{principle of minimum power} for some
boundary potential $\psi \in \R^T$ if $\phi$ minimizes the extended power
functional $P$ subject to the constraint that $\phi|_T = \psi$.
\end{defn}

The next result says that given a boundary potential, there is a unique compatible boundary current. This boundary current can be computed using the extended power functional.

\begin{prop}
  Let $P$ be the extended power functional for an open circuit. Then 
  \begin{align*}
    Q\maps \R^T &\longrightarrow \R \\
    \psi &\longmapsto \min_{\phi|_T=\psi }  P(\phi). 
  \end{align*}
  is a well-defined function. Moreover, given a boundary potential $\psi \in
  \R^T$, the gradient $\nabla Q_\psi \in \R^T$
  \begin{align*}
    \nabla Q_\psi: T &\longrightarrow \R \\
    n &\longmapsto \frac{\partial Q}{\partial \phi_n}\bigg\vert_{\phi =\psi}
  \end{align*}
  is the unique boundary current $\iota$ such that
  \begin{itemize}
    \item $\psi$ extends to a potential $\phi \in \R^N$,
    \item $\iota$ extends to a current $\overline\iota \in \R^N$,
    \item $\phi$, $\overline\iota$ obey Ohm's law,
    \item $\overline\iota$ obeys Kirchhoff's current law at all $n \in N-T$.
  \end{itemize}
  We call this function $Q$ the \define{power functional} of the open circuit.
\end{prop}
\begin{proof}
  See \cite[Propositions 2.6 and 2.10]{BaezFong}.
\end{proof}

This implies that the set of physically realizable boundary potentials and boundary currents for the open circuit can be written as
\[
  \mathrm{Graph}(\nabla Q) = \big\{(\psi, \nabla Q_\psi) \, \vert \, \psi \in
  \R^T\big\} \subseteq \R^T \oplus \R^T.
\]

It remains to understand how boundary potentials and currents are related to
potentials and currents on the inputs and outputs of an open circuit.  This has
three key aspects. Recall that we think of the inputs and outputs as points at
which we can connect our open circuit with another open circuit. The first
principle is then that the potential at some input $x \in X$ or output $y \in Y$
is equal to the potential at its corresponding terminal, $i(x)$ or $o(y)$
respectively. Thus if $[i,o]\maps X+Y \to T$ is the function describing both the
input and output maps, then the boundary potential $\psi \in \R^T$ is realized
by the potential $\psi' = \psi \circ [i,o]$.

The second principle is that the current outflow at a terminal is split among
the inputs and outputs that connect to it. The final principle is that the
quantity associated to inputs is net current \emph{inflow}, which is the
negative of net current outflow. Thus if $\iota$ is a boundary current, then it
may be realized by any current $\iota' \in \R^X \oplus \R^Y$ such that
\[
  \iota(n) = \sum_{x \in o^{-1}(n)} \iota'(x)- \sum_{x \in i^{-1}(n)}\iota'(x).
\]

These principles define a linear relation 
\[
  S[i,o]\maps \R^T \oplus \R^T \leadsto \R^X \oplus \R^X \oplus \R^Y \oplus
  \R^Y
\]
by which we mean a linear subspace of $ \R^T \oplus \R^T \oplus \R^X \oplus \R^X \oplus \R^Y \oplus \R^Y$. Applying this relation pointwise to a subspace of $\R^T \oplus \R^T$ allows us to view $S[i,o]$ also as a function
\[
  \Big\{\textrm{subspaces of } \R^T \oplus \R^T\Big\} \longrightarrow 
  \Big\{\textrm{subspaces of }\R^X \oplus \R^X \oplus \R^Y \oplus \R^Y \Big\}.
\]
This function takes a subspace of $\R^T \oplus \R^T$, interpreted as a
collection of` boundary potential--boundary current pairs, to the subspace of
all compatible potentials and currents on the inputs and outputs. This lets us
give an explicit definition of the behavior of an open circuit:

\begin{defn} \label{defn:behavior}
The \define{behavior} of an open circuit with power functional $Q$ is given by
the subspace
\[
  S[i,o]\big(\mathrm{Graph}(\nabla Q)\big)
\]
of $\R^X \oplus \R^X \oplus \R^Y \oplus \R^Y$.
\end{defn}

\section{The principle of minimum dissipation}
\label{sec:dissipation}

In this section, we analyze an open detailed balanced Markov process $M \maps
(X,u) \to (Y,v)$ using ideas from circuit theory.  So, assume that we have a
cospan of finite sets with populations
\[ 
  \xymatrix{ 
    & (N,q) & \\
    (X,u) \ar[ur]^{i} && (Y,v) \ar[ul]_{o} 
  }
\]
together with a detailed balanced open Markov process on $N$:
\[ 
  \xymatrix{  (0,\infty) & E \ar[l]_-r \ar[r]<-.5ex>_t \ar[r] <.5ex>^s & N \ar[r]^-q & (0,\infty) } .
\]
The master equation is closely related to a quadratic form analogous to the
extended power functional for electrical circuits.

Just as we take potentials as our starting point for the study of circuits, we take
`deviations' as our starting point for detailed balanced Markov processes. The \define{deviation} of a population $p$ on a detailed balanced Markov process with 
chosen equilibrium population $q$ is the ratio between $p$ and $q$:
\[    x_i = \frac{p_i}{q_i}   .\]

\begin{defn} The \define{extended dissipation functional} of a detailed balanced open Markov process 
is the quadratic form $C \maps \R^N \to \R$ given by
\[ C(p) = \frac{1}{4} \sum_{e \in E} r_e q_{s(e)} \biggl( \frac{p_{s(e)}}{q_{s(e)}} - \frac{p_{t(e)}}{q_{t(e)}} \biggr)^2.
\]
\end{defn}

At the end of Section \ref{sec:symplectic} we shall see that the master equation describes
the gradient flow associated to this functional.  Thus, the population moves `downhill', toward lower values of $C$.  However, this gradient must be computed
using a suitable metric on $\R^N$.  For now we state this result in less geometric
language:

\begin{thm} \label{thm:main}
  Fix a state $n \in N$ and suppose that a time-dependent population $p(t) \maps N \to (0,\infty)$ obeys the master equation at $n$. Then
  \[
    \frac{dp_n}{dt} = -q_n \frac{\partial C}{\partial p_n}.
  \]
\end{thm}

\begin{proof}
Our calculation considers how population flows between each state $m \in N$ and
our given state $n$. Differentiating $C$ with respect to $p_n$ yields
\[ 
  \frac{ \partial C}{\partial p_n} =  \sum_{m \in N}\left[ 
\sum_{e \maps n \to m} \frac{r_e}2  \left( \frac{p_n}{q_n}-\frac{p_m}{q_m} \right) - \sum_{e \maps m\to n} \frac{r_e q_m}{2q_n} \left( \frac{p_m}{q_m}-\frac{p_n}{q_n} \right) \right]. 
\]
If we then multiply by $-q_n$ and factor the result we arrive at
\[ 
  -q_n\frac{ \partial C}{\partial p_n} = \sum_{m \in N}\left[\frac12\bigg(
 \sum_{e \maps n \to m} r_eq_n + \sum_{e\maps m\to n} r_e q_m\bigg)  \left( \frac{p_m}{q_m}-\frac{p_n}{q_n}
\right)\right]. 
\]
We shall show that each term in this sum over $m$ represents the net flow
from $m$ to $n$, and hence that the sum represents the net inflow at $n$.

Detailed balance says that the two summands in the first parentheses are
equal. We may thus rewrite the coefficient of $\frac{p_m}{q_m}-\frac{p_n}{q_n}$
in either of two ways:
\[
\frac12\bigg(\sum_{e\maps m\to n} r_e q_m +\sum_{e \maps n \to m} r_eq_n\bigg) =
\sum_{e\maps m\to n} r_e q_m = \sum_{e \maps n \to m} r_eq_n.
\]
We thus obtain 
\[ 
  -q_n\frac{ \partial C}{\partial p_n} = \sum_{m \in N}\left[\bigg(\sum_{e\maps
  m\to n} r_e q_m\bigg) \frac{p_m}{q_m} - \bigg(\sum_{e \maps n \to m} r_eq_n\bigg) \frac{p_n}{q_n} \right]. 
\]
or simply
\[ 
  -q_n\frac{ \partial C}{\partial p_n} = \sum_{m \in N}\left[\sum_{e\maps
  m\to n} r_e p_m -\sum_{e \maps n \to m} r_ep_n\right] = \frac{dp_n}{dt}
\]
where in the last step we use the master equation.
\end{proof}

In particular, note that $\frac{dp_i}{dt} = 0$ if and only if $\frac{\partial C}{\partial p_i}
=0$. Thus a population $p$ is a steady state solution of the open master
equation if and only if $\frac{\partial C}{\partial p_i} = 0$ at all internal 
states $i$. As the dissipation is a positive definite quadratic form, this implies that
$p$ is a steady state if and only if it minimizes $C$ subject to the constraint
that the boundary population is equal to the restriction of $p$ to the boundary.

\begin{defn}
We say a population $p \in \R^N$ obeys the \define{principle of minimum
dissipation} for some boundary population $b \in \R^T$ if $p$ minimizes the
extended dissipation functional $C(p)$ subject to the constraint that $p|_T = b$.
\end{defn}

This lets us state a corollary of Theorem \ref{thm:main}:

\begin{cor}
Given an open detailed balanced Markov process $M \maps X \to Y$ with chosen
equilibrium $q\in \R^N$, a population $p \in \R^N$ obeys the principle of
minimum dissipation for some boundary population $b \in \R^T$ if and only if
$p$ satisfies the open master equation with boundary
conditions $b$.
\end{cor}

\begin{defn}
Given an open detailed balanced Markov process with terminals $T$, the \define{dissipation functional} $D \maps \R^T \to \R$ is given by
\[ D(b) = \min_{p|_T = b} C(p)
\]
where we minimize over all populations $p \in \R^N$ that restrict to 
$b$ on the terminals.
\end{defn}

Like the power functional for circuits, this functional lets us compute the
relation an open detailed balanced Markov process imposes on population--flow
pairs at the terminals. We can then compute the allowed population--flow pairs
on the inputs and outputs with the help of the map $S[i,o]$, which we already
used for the same purpose in the case of circuits.  This results in the
following expression of the behavior of an open detailed balanced Markov
process:

\begin{defn}
The \define{behavior} of an open detailed balanced Markov process with
dissipation functional $D$ is given by the subspace
\[
  S[i,o]\big(\mathrm{Graph}(\nabla D)\big)
\]
of $\R^X \oplus \R^X \oplus \R^Y \oplus \R^Y$.
\end{defn}

\section{From detailed balanced Markov processes to circuits} 
\label{sec:reduction}

Comparing the dissipation of a detailed balanced open
Markov process:
\[ 
  C(p) = \frac{1}{4} \sum_{e \in E} r_e q_{s(e)} \bigg(\frac{p_{s(e)}}{q_{s(e)}} -
  \frac{p_{t(e)}}{q_{t(e)}} \bigg)^2
\] 
to the extended power functional of a circuit:
\[ P(\phi) = \frac{1}{2} \sum_{e \in E} c_e ( \phi_{s(e)} - \phi_{t(e)} )^2 \]
suggests the following correspondence:
\[ \begin{array}{ccc} 
\displaystyle{ \frac{p_i}{q_i} } &\leftrightarrow & \phi_i \\ \\
\displaystyle{ \frac{1}{2} r_e q_{s(e)}} &\leftrightarrow & c_e . 
\end{array} \]
This sharpens this analogy between detailed balanced Markov processes and circuits made of resistors.   In this analogy, the population of a state 
is \emph{roughly} analogous to the electric potential at a node.  However, it is really the `deviation', the ratio of the population to the equilibrium population, that 
is analogous to the electric potential.  Similarly, the rate constant of an edge is \emph{roughly}  analogous to the conductance of an edge.  However, it is really half the rate constant  times the equilibrium population at the source of that edge that is analogous to the conductance.  The factor of $1/2$ compensates for the fact that in a circuit each edge allows for flow in both directions, while in a Markov process each edge allows for flows in only one direction.

We thus shall convert an open detailed balanced Markov process $M \maps X \to Y$, namely:
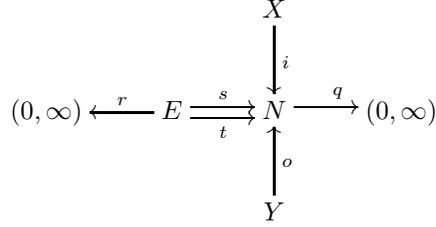
\[ 
  \xymatrix{ 
  && X \ar[d]^{i} \\
(0,\infty) & E \ar[l]_-r \ar[r]<-.5ex>_t \ar[r] <.5ex>^s & N \ar[r] <.5ex>^q &
(0,\infty) \\ 
&& Y \ar[u]_{o}} 
\]
into an open circuit $K(M) \maps X \to Y$, namely:
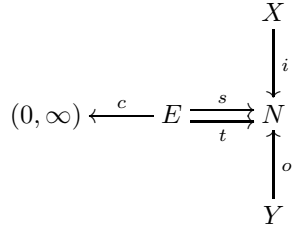
\[ 
  \xymatrix{  
    && X \ar[d]^{i} \\
    (0,\infty) & E \ar[l]_-{c} \ar[r]<-.5ex>_t  \ar[r] <.5ex>^s & N  \\
    && Y \ar[u]_{o}
  }
\]
where 
\[  c_e = \frac{1}{2} r_e q_{s(e)}. \]
For the open detailed balanced Markov process with two states depicted below this
map $K$ has the following effect:
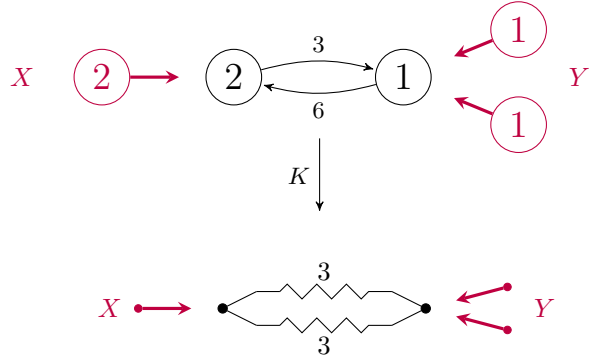
\[ 
\begin{tikzpicture}[circuit ee IEC, set resistor graphic=var resistor IEC graphic,->,>=stealth',shorten >=1pt,]
  \node[main node](1) {$2$};
  \node[main node](2) [right=1.5cm of 1] {$1$};

\node(A) [terminal, left=1cm of 1] {$2$};
\node(B) [terminal, above right=0.1cm and 1cm of 2] {$1$};
\node(C) [terminal, below right=0.1cm and 1cm of 2] {$1$};
\node[right=1.7 cm of 2,color=purple] {$Y$};
\node[left=0.4 cm of A,color=purple] {$X$};
  \path[every node/.style={font=\sffamily\small}, shorten >=1pt]
    (1) edge [bend left=15] node[above] {$3$}(2) 
    (2) edge [bend left=15] node[below] {$6$} (1);
    
\path[color=purple, very thick, shorten >=10pt, ->, >=stealth] (A) edge (1);
\path[color=purple, very thick, shorten >=10pt, ->, >=stealth] (B) edge (2);
\path[color=purple, very thick, shorten >=10pt, ->, >=stealth] (C) edge (2);

\node(D) [below right=.3 cm and 0.75cm of 1] {}; 
\node(E) [below=1cm of D] {};
  \path[every node/.style={font=\sffamily\small}, shorten >=1pt]
  (D) edge node[left] {$K$} (E);
\end{tikzpicture} 
\]
\[
\hskip 2em 
\begin{tikzpicture}[circuit ee IEC, set resistor graphic=var resistor IEC
    graphic,scale=.9]
    \node[contact]         (A) at (0,0) {};
    \node[contact]         (B) at (3,0) {};
    
        \coordinate         (ua) at (.5,.25) {};
    \coordinate         (ub) at (2.5,.25) {};
    \coordinate         (la) at (.5,-.25) {};
    \coordinate         (lb) at (2.5,-.25) {};
    \path (A) edge (ua);
    \path (A) edge (la);
    \path (B) edge (ub);
    \path (B) edge (lb);
    \path (ua) edge  [resistor] node[label={[label distance=1pt]90:{$3$}}] {} (ub);
    \path (la) edge  [resistor] node[label={[label distance=1pt]270:{$3$}}] {} (lb);
    
    \node(L) [left=1cm of A,circle,draw,inner sep=1pt,fill=gray,color=purple] {};     
\path[color=purple, very thick, shorten >=10pt, ->, >=stealth] (L) edge (A);  
\node[left=1.15 cm of A,color=purple] {$X$};
\node(R) [above right=0.2cm and 1cm of B,circle,draw,inner sep=1pt,fill=gray,color=purple] {};
\node(R2) [below right=0.2cm and 1cm of B,circle,draw,inner sep=1pt,fill=gray,color=purple] {};
\path[color=purple, very thick, shorten >=10pt, ->, >=stealth] (R) edge (B);
\path[color=purple, very thick, shorten >=10pt, ->, >=stealth] (R2) edge (B);
\node[right=1.25 cm of B,color=purple] {$Y$};
   \end{tikzpicture} 
\]

This analogy is stronger than a mere visual resemblance. The behavior for the Markov process $M$ is easily obtained from the behavior of the circuit $K(M)$.  Indeed, write $C_M$ for the extended dissipation functional of the open detailed balanced Markov process $M$, and $P_{K(M)}$ for the extended power functional of the open circuit $K(M)$. Then
\[
  C_M(p) = P_{K(M)}(\tfrac{p}{q}).
\]
Minimizing over the interior, we also have the equivalent fact for the
dissipation functional $D_M$ and power functional $Q_{K(M)}$:
\begin{lem}
  Let $M$ be an open detailed balanced Markov process, and let $K(M)$ be the
  corresponding open circuit. Then
\[
  D_M(p) = Q_{K(M)}(\tfrac{p}{q}),
\]
where $D_M$ is the dissipation functional for $M$ and $Q_{K(M)}$ is the power
functional for $K(M)$.
\end{lem}

Consider now $\Graph(\nabla D_M)$ and $\Graph(\nabla Q_{K(M)})$. These are both
subspaces of $\R^T \oplus \R^T$. For any set with populations $(T,q)$, define
the function
\begin{align*}  
\alpha_{T,q}\maps \R^T \oplus \R^T &\longrightarrow \R^T \oplus \R^T; \\
(\phi,\iota) &\longmapsto (q\phi,\iota).
\end{align*}
Then if $T$ is the set of terminals and $q\maps T \to (0,\infty)$ is the
restriction of the population function $q\maps N \to (0,\infty)$ to the
terminals, we see that 
\[
  \Graph(\nabla D_M) = \alpha_{T,q}(\Graph(\nabla Q_{K(M)})).
\]
Note here that we are applying $\alpha_{T,q}$ pointwise to the subspace
$\Graph(\nabla Q_{K(M)})$ to arrive at the subspace $\Graph(\nabla D_M)$.

Observe that $\alpha_{T,q}$ acts as the identity on the `current' or `flow'
summand, while $S[i,o]$ acts simply as precomposition by $[i,o]$ on the
`potential' or `population' summand. This implies the equality of the composite
relations
\[
  S[i,o] \alpha_{T,q} = (\alpha_{X,qi} \oplus \alpha_{Y,qo})S[i,o]
\]
as relations 
\[
  \R^T \oplus \R^T \leadsto \R^X \oplus \R^X \oplus \R^Y \oplus \R^Y.
\]
In summary, we have arrived at the following theorem.
\begin{thm} \label{thm.behaviors}
  Let $M$ be an open detailed balanced Markov process, and let $K(M)$ be the
  corresponding open circuit. Then
\[
  S[i,o]\Graph(\nabla D_M) = \big(\alpha_{X,qi} \oplus
  \alpha_{Y,qo}\big)S[i,o]\big(\Graph(\nabla Q_{K(M)})\big).
\]
where $D_M$ is the dissipation functional for $M$ and $Q_{K(M)}$ is the power
functional for $K(M)$.
\end{thm}
This makes precise the relationship between the two behaviors. Observe that
population deviation is analogous to electric potential, and population
flow is analogous to electric current.

In fact, this analogy holds not just for open detailed balanced Markov processes
and open circuits in isolation, but continues to hold when we build up
these Markov processes and circuits from subsystems. For this we need to discuss
what it means to compose open systems.

\section{Composing open Markov processes} 
\label{sec:composing}

If the outputs of one open system match the inputs of another we should be able to glue them together, or `compose' them, and obtain a new open system.  For open Markov processes, composition is fairly intuitive.  Suppose we have an open Markov process $M \maps X \to Y$:

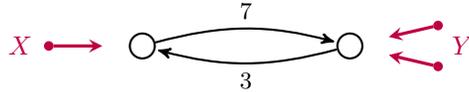
\[
\begin{tikzpicture}[->,>=stealth',shorten >=1pt,auto,node distance=3.7cm,
  thick,main node/.style={circle,fill=white!20,draw,font=\sffamily\Large\bfseries},terminal/.style={circle,fill=blue!20,draw,font=\sffamily\Large\bfseries}]]
  \node[main node](1) {$$};
  \node[main node](2) [right=2.4cm of 1] {$$};

\node(A) [left=1cm of 1,circle,draw,inner sep=1pt,fill=gray,color=purple] {};
\node(B) [above right=0.1cm and 1cm of 2,circle,draw,inner sep=1pt,fill=gray,color=purple] {};
\node(C) [below right=0.1cm and 1cm of 2,circle,draw,inner sep=1pt,fill=gray,color=purple] {};
\node[right=1.05 cm of 2,color=purple] {$Y$};
\node[left=1.15 cm of 1,color=purple] {$X$};
  \path[every node/.style={font=\sffamily\small}, shorten >=1pt]
    (1) edge [bend left=15] node[above] {$7$}(2) 
    (2) edge [bend left=15] node[below] {$3$} (1);
    
\path[color=purple, very thick, shorten >=10pt, ->, >=stealth] (A) edge (1);
\path[color=purple, very thick, shorten >=10pt, ->, >=stealth] (B) edge (2);
\path[color=purple, very thick, shorten >=10pt, ->, >=stealth] (C) edge (2);
\end{tikzpicture}
\]
and an open Markov process $M' \maps Y \to Z$:
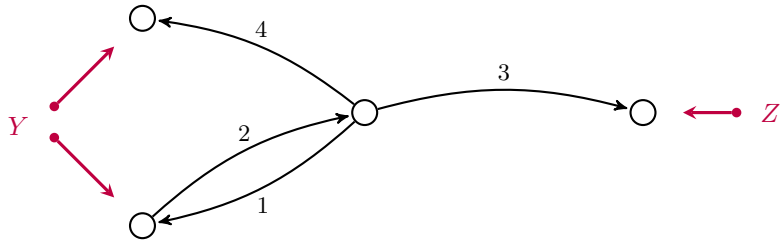
\[
\begin{tikzpicture}[->,>=stealth',shorten >=1pt,auto,node distance=3.7cm,
  thick,main node/.style={circle,fill=white!20,draw,font=\sffamily\Large\bfseries},terminal/.style={circle,fill=blue!20,draw,font=\sffamily\Large\bfseries}]]
  \node[main node](1) {$$};
  \node[main node](2) [below=2.4cm of 1] {$$};
  \node[main node](3) [below right=1cm and 2.7cm of 1]  {$$};
  \node[main node](4) [right of=3] {$$};
\node(A) [below left=1cm and 1cm of 1,circle,draw,inner sep=1pt,fill=gray,color=purple] {};
\node(B) [above left=1cm and 1cm of 2,circle,draw,inner sep=1pt,fill=gray,color=purple] {};
\node(C) [right=1cm of 4,circle,draw,inner sep=1pt,fill=gray,color=purple] {};
\node[below left=1.05 cm and 1.25 cm of 1,color=purple] {$Y$};
\node[right=1.25 cm of 4,color=purple] {$Z$};
  \path[every node/.style={font=\sffamily\small}, shorten >=1pt]
    (3) edge [bend left =15] node[above] {$3$} (4)
    (2) edge [bend left=15] node[above,midway] {$2$} (3) 
    (3) edge [bend left =15] node[below] {$1$} (2)
    (3) edge [bend right =15] node[above] {$4$} (1);
    
\path[color=purple, very thick, shorten >=10pt, ->, >=stealth] (C) edge (4);
\path[color=purple, very thick, shorten >=10pt, ->, >=stealth] (A) edge (1);
\path[color=purple, very thick, shorten >=10pt, ->, >=stealth] (B) edge (2);
\end{tikzpicture}
\]
Then we can compose them by gluing the outputs of the first to the inputs of
the second:
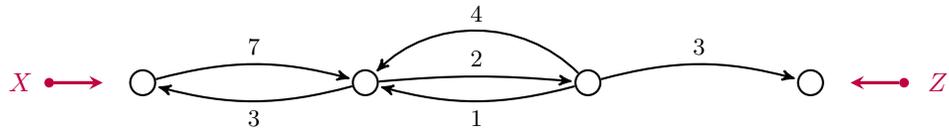
\[
\begin{tikzpicture}[->,>=stealth',shorten >=1pt,auto,node distance=3.7cm,
  thick,main node/.style={circle,fill=white!20,draw,font=\sffamily\Large\bfseries},terminal/.style={circle,fill=blue!20,draw,font=\sffamily\Large\bfseries}]]

  \node[main node](1) {$$};
  \node[main node](2) [right=2.6cm of 1] {$$};
  \node[main node](3) [right=2.6cm of 2] {$$};
  \node[main node](4) [right=2.6cm of 3] {$$};
  
    \path[every node/.style={font=\sffamily\small}, shorten >=1pt]
    (1) edge [bend left=15] node[above] {$7$}(2) 
    (2) edge [bend left=15] node[below] {$3$} (1)
        (3) edge [bend left =15] node[above] {$3$} (4)
    (2) edge [bend left=5] node[above,midway] {$2$} (3) 
    (3) edge [bend left =15] node[below] {$1$} (2)
    (3) edge [bend right =45] node[above] {$4$} (2);
 
\node(A) [left=1cm of 1,circle,draw,inner sep=1pt,fill=gray,color=purple] {};     
\path[color=purple, very thick, shorten >=10pt, ->, >=stealth] (A) edge (1);  
\node[left=1.15 cm of 1,color=purple] {$X$};
\node(C) [right=1cm of 4,circle,draw,inner sep=1pt,fill=gray,color=purple] {};
\path[color=purple, very thick, shorten >=10pt, ->, >=stealth] (C) edge (4);
\node[right=1.25 cm of 4,color=purple] {$Z$};

  \end{tikzpicture} \]
Note that in this example the output map of $M$ sends both outputs to the same state. This is why when we compose the two Markov processes we identify the two inputs of $M'$.

Decorated cospans let us formalize this process of composition, not only for
open Markov process but for two other kinds of open systems we need in this paper: detailed balanced Markov processes and circuits made of resistors.  Recall that an open Markov process $M \maps X \to Y$ is a cospan of finite sets:
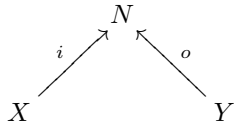
\[ \xymatrix{ & N & \\
X \ar[ur]^{i} && Y \ar[ul]_{o} \\ } \]
together with a Markov process $M$ on $N$.  We say that the cospan is \define{decorated}
with the Markov process $M$.  To compose open Markov processes, we need to 
compose the cospans and also compose these decorations.
  
Composing two cospans
\[ \xymatrix{ & N & & N' & \\ 
X \ar[ur]^{i} & &  Y \ar[ul]_{o} \ar[ur]^{i'} & & Z \ar[ul]_{o'} \\ } \] 
yields a new cospan
\[ \xymatrix { & N+_{Y}N' & \\
 X \ar[ur]^{j i}  & & Z. \ar[ul]_{j' o'} \\ }\]
Here $N+_{Y}N'$ is the `pushout' of $N$ and $N'$ over $Y$, which comes with canonical
maps written $j\maps N \to N+_{Y}N'$ and $j'\maps N' \to N+_{Y}N'$.  When $N$
and $N'$ are sets, we construct this pushout by taking the disjoint union of $N$
and $N'$ and quotienting by the equivalence relation where $n \sim n'$ if $o(y)
= n$ and $i'(y) = n'$ for some $y \in Y$.  This identifies the outputs of the
first open Markov process with the inputs of the second.  

Besides composing the cospans, we need to compose the Markov processes decorating them. We do this with the help of two constructions.  The first  says how to `push forward' a Markov process on $N$ to a Markov process on $N'$ given a map $f \maps N \to N'$.   To explain this, it is useful to let $F(N)$ denote the set of all Markov processes on $N$.  

\begin{lem} \label{functor}
Given any function $f \maps N \to N'$ there is a function $F(f) \maps F(N) \to F(N')$ that maps this Markov process on $N$:
\[ \xymatrix{ (0,\infty) & E \ar[l]_-r \ar[r]<-.5ex>_t \ar[r] <.5ex>^s & N  } \]
to this Markov process on $N'$:
\[ \xymatrix{ (0,\infty) & E \ar[l]_-r \ar[r]<-.5ex>_{f t} \ar[r] <.5ex>^{f s} & N'  }.\]
Moreover, we have 
\[          F(fg) = F(f) F(g)  \]
and 
\[          F(1_X) = 1_{F{X}}  \]
when $1_X \maps X \to X$ is the identity map.
\end{lem}

\begin{proof}  These are easy to check.  While $F(N)$ is actually a proper class in general, we can treat it as a `large' set using Grothendieck's axiom of universes, so no contradictions arise \cite{Fong}.  Note that this lemma actually says $F \maps \FinSet \to \Set$ is a functor, where $\Set$ is the category of large sets.
\end{proof}

The second construction says how to take a Markov process on $N$ and a
Markov process on $N'$ and get a Markov process on the disjoint union of $N$ and 
$N'$.   For this we need a bit of notation:

\begin{itemize}
\item Given two sets $X$ and $X'$, let $X + X'$ stand for the disjoint union of $X$ and
$X'$.
\item Given two functions $g \maps X \to Y$ and $g' \maps X' \to Y$, let $[g,g'] \maps X + X' \to Y$ be the unique function that restricts to $g$ on $X$ and $g'$ on $X'$.
\item Given two functions $g \maps X \to Y$ and $g' \maps X' \to Y'$, let $g+g' \maps X + X' \to Y + Y'$ be the unique function that restricts to $g$ on $X$ and $g'$ on $X'$.
\end{itemize}

\begin{lem} \label{lax}
For any pair of finite sets $N$ and $N'$, 
there is a map $\phi_{N,N'} \maps F(N) \times F(N') \to F(N + N')$ sending any pair consisting of a Markov process on $N$:
\[ \xymatrix{ (0,\infty) & E \ar[l]_-r \ar[r]<-.5ex>_t \ar[r] <.5ex>^s & N  } \]
and a Markov process on $N'$:
\[   \xymatrix{ (0,\infty) & E' \ar[l]_-{r'} \ar[r]<-.5ex>_{t'} \ar[r] <.5ex>^{s'} & N'  } 
\]
to this Markov process on $N+N'$:
\[    \xymatrix{ (0,\infty) & E+E' \ar[l]_{[r,r']} \ar[r]<-.5ex>_{t+t'} \ar[r] <.5ex>^{s+s'} & N+N' }.\] 
This map makes $F \maps \FinSet \to \Set$ into a lax monoidal functor.
\end{lem}

\begin{proof}
Saying that $\phi$ makes $F$ into a lax monoidal functor means that together
with some map $\phi_1$ from the one-element set to $F(\emptyset)$, it makes
the three diagrams of Definition \ref{defn.lmf} commute. There is only one possible
choice of $\phi_1$, and it is easy to check that the diagrams commute.
\end{proof}

Now, suppose we have an open Markov process $M \maps X \to Y$ and 
and an open Markov process $M' \maps Y \to Z$.  Thus, we have cospans of
finite sets
\[ \xymatrix{ & N & & N' & \\ 
X \ar[ur]^{i} & &  Y \ar[ul]_{o} \ar[ur]^{i'} & & Z \ar[ul]_{o'} \\ } \] 
decorated with Markov processes $M \in F(N)$, $M' \in F(N')$.   To define the \define{composite} open Markov process $M' M \maps X \to Z$, we need to
decorate the composite cospan
\[ \xymatrix { & N+_{Y}N' & \\
 X \ar[ur]^{j \circ i}  & & Z. \ar[ul]_{j' \circ o'} \\ }\] 
with a Markov process.  Thus, we need to choose an element of $F(N+_Y N')$.  
To do this, we take $\phi_{N,N'}(M,M')$, which is a Markov process on $N + N'$, and
apply the map 
\[        F([j,j'])  \maps F(N + N') \to F(N+_Y N').  \]
This gives the required Markov process on $N +_Y N'$.

With this way of composing open Markov processes, we \emph{almost} get a category with finite sets as objects and open Markov processes $M \maps X \to Y$ as morphisms from $X$ to $Y$.  However, the disjoint union of open sets is associative only up to isomorphism, and thus so is the pushout, and so is composition of open Markov
processes.  Thus, we obtain a subtler structure, called a bicategory
\cite{Benabou,Leinster}.  While we expect this to be important eventually, bicategories 
are a bit distracting from our current goals.  Thus, we shall take \emph{isomorphism classes} of open Markov processes $M \maps X \to Y$ as morphisms from $X$ to $Y$, 
obtaining a category.   For the details of what kind
of `isomorphism class' we mean here, see Appendix \ref{sec:decorated}.

\begin{defn}
The category $\Mark$ is the decorated cospan category where an object is a finite set and a morphism is an isomorphism class of open Markov processes $M \maps X \to Y$. 
\end{defn}

In fact $\Mark$ is much better than a mere category.  First, it is a `symmetric monoidal'
category, meaning that it has a well-behaved tensor product, which describes our ability to take the `disjoint union' of open Markov processes $M \maps X \to Y$, $M' \maps X'
\to Y'$ and get an open Markov process $M+M' \maps X+X' \to Y+Y'$.  Second, it is a `dagger category', meaning that we can turn around an open Markov process $M \maps X \to Y$ and regard it as an open Markov process $M^\dagger \maps Y \to X$.  Third, these features fit together in a nice way, giving a `dagger compact' category.  For an introduction to these concepts see \cite{AC,BaezStay,Se}.   We will not need these facts about the category $\Mark$ here, but they mean that it fits into a general program relating categories to diagrams of networks. 

\begin{lem} \label{Mark_dagger}
The category $\Mark$ is a dagger compact category.
\end{lem}

\begin{proof}
We use Lemmas \ref{functor} and \ref{lax}, which say that $F \maps
\FinSet \to \Set$ is a lax monoidal functor, together with the further fact that $F$ is a 
lax symmetric monoidal functor.  The conclusion then follows from Lemma \ref{dagger_lemma}. 
\end{proof}

\section{Composing open detailed balanced Markov processes}
\label{sec:composing_detailed}

Just as open Markov processes form a decorated cospan category, so too do open
detailed balanced Markov processes. The intuition is the same: we compose two
open detailed balanced Markov processes by identifying, or `gluing together', the outputs of the first and the inputs of the second. There is, however, one important difference: for a detailed balanced Markov processes, each state is equipped with an equilibrium population. Thus to glue together two such processes, we require that the equilibrium populations at the outputs of the first process match those at the inputs of the second.

Why do we require this?  If we take two open detailed balanced Markov processes, regard them as mere open Markov processes, and compose them as such, the resulting open Markov process may not have a detailed balanced equilibrium!  Avoiding this problem is one reason we defined a detailed balanced Markov process in the way we did.  It is not merely a Markov process for which a detailed balanced equilibium \emph{exists}, but a Markov process \emph{equipped with} a specific detailed balanced equilibrum.  This allows us to impose the matching condition just described---and then, open detailed balanced Markov processes are closed under composition.  They become the morphisms of a category, which we call $\DetBalMark$.

We now construct this category.  Let $\FinPopSet$ be the category of finite sets with populations, as in Definition \ref{defn:FinPopSet}. A morphism in $\DetBalMark$ will be a cospan in $\FinPopSet$ decorated by a detailed balanced Markov process.  To compose such cospans, we need $\FinPopSet$ to have pushouts.  Pushouts are a special case of
finite colimits.

\begin{lem}
  The category $\FinPopSet$ has finite colimits. Moreover, the forgetful
  functor $U\maps \FinPopSet \to \FinSet$, that maps each finite set with
  populations to its underlying set, preserves finite colimits.
\end{lem}
\begin{proof}
  This follows from properties of so-called slice categories; details may be
  found in any basic category theory textbook, such as Mac Lane \cite{MacLane}.
  Here we provide a proof of this special case.

  Let $D\maps J \to \FinPopSet$ be a finite diagram in $\FinPopSet$. Composing
  with the forgetful functor $U\maps \FinPopSet \to \FinSet$, we have a finite
  diagram $U \circ D$ in $\FinSet$.  Write $A$ for the colimit of $U \circ D$;
  this exists as $\FinSet$ has finite colimits. 

  The inclusion functor $I\maps \FinSet \to \Set$ preserves colimits. 
  Thus the finite set $A$ is also the colimit of $I \circ U
  \circ D$ in $\Set$. The key to this lemma is then to note that each object in
  the image of $D$ is a finite set with population, and the population
  assignments form a cocone of the diagram $I \circ U \circ D$ in $\Set$. By the
  universal property of $A$ in $\Set$, we thus obtain a map $q\maps A \to
  (0,\infty)$. It is then easy to check that $(A,q)$ is the colimit of $D$ in
  $\FinPopSet$.

  Moreover, as $U$ maps $(A,q)$ to $A$, we see that $U\maps \FinPopSet \to
  \FinSet$ preserves finite colimits.
\end{proof}

Copying the argument for open Markov processes in Section \ref{sec:composing},
we can now define a lax monoidal functor $G \maps (\FinPopSet,+) \to
(\Set,\times)$ sending any finite set $X$ with populations to the set of
detailed balanced Markov processes on $X$.  Given any finite set with populations 
$(N,q)$, let $G(N,q)$ be the set of all detailed balanced Markov processes on $(N,q)$.

\begin{lem}
  Given any map of sets with populations $f \maps (N,q) \to (N',q')$, there is a
  function $G(f) \maps G(N,q) \to G(N',q')$ mapping any detailed balanced Markov
  process $M$ on $N$:
  \[ \xymatrix{ (0,\infty) & E \ar[l]_-r \ar[r]<-.5ex>_t \ar[r] <.5ex>^s & N
  \ar[r] <.5ex>^q & (0,\infty) } \]
  to this detailed balanced Markov process on $N'$:
  \[ \xymatrix{ (0,\infty) & E \ar[l]_-r \ar[r]<-.5ex>_{ft} \ar[r] <.5ex>^{fs} &
  N' \ar[r] <.5ex>^{q'} & (0,\infty). } \]
	
\end{lem}
\begin{proof}
  We need to check that $G(f)(M)$ is detailed balanced. This requires that 
  \[      
    \sum_{e \maps n \to m} r_e q'_n = \sum_{e \maps m \to n} r_e q'_m 
  \] 
  for all $n,m \in N'$. To see this, we note that every edge between $n$ and $m$
  in $G(f)(M)$ is the image of an edge between $i$ and $j$ in $M$ for some $i$,
  $j$ such that $f(i)=n$ and $f(j)=m$. Moreover, as $M$ is detailed balanced at
  population $q = q'f$, the net flow from $i$ to $j$ is equal to the net flow
  from $j$ to $i$ in $M$.  Explicitly: 
  \begin{multline*}
    \sum_{e \maps n \to m} r_e q'_n = \sum_{\substack{i \in f^{-1}(n)\\j \in
    f^{-1}(m)}} \Bigg(\sum_{e \maps i \to j} r_e q'_n \Bigg)
= \sum_{\substack{i \in f^{-1}(n)\\j \in
f^{-1}(m)}} \Bigg(\sum_{e \maps i \to j} r_e q_i \Bigg)\\
= \sum_{\substack{i \in f^{-1}(n)\\j \in
f^{-1}(m)}} \Bigg(\sum_{e \maps j \to i} r_e q_j\Bigg)
= \sum_{e \maps m \to n} r_e q'_m
  \end{multline*} 
Here $e\maps i \to j$ means that $e$ is an edge from $i$ to $j$ in $M$, whereas $e\maps n \to m$ denotes that $e$ is an edge from $n$ to $m$ in $G(f)(M)$. 
\end{proof}

Adding coherence maps then gives a lax monoidal functor:

\begin{lem}
 For each pair $(N,q), (N',q')$ of finite sets with populations, there is a map from $G(N,q)
 \times G(N',q')$ to $G(N + N',[q,q'])$ sending any pair consisting of a detailed balanced
 Markov process on $(N,q)$:
\[ \xymatrix{ (0,\infty) & E \ar[l]_-r \ar[r]<-.5ex>_t \ar[r] <.5ex>^s & N \ar[r] <.5ex>^q & (0,\infty)  } \]
and a detailed balanced Markov process on $(N',q')$:
\[   \xymatrix{ (0,\infty) & E' \ar[l]_-{r'} \ar[r]<-.5ex>_{t'} \ar[r]
<.5ex>^{s'} & N' \ar[r] <.5ex>^{q'} & (0,\infty)  } 
\]
to this detailed balanced Markov process on $(N + N',[q,q'])$:
\[    \xymatrix{ (0,\infty) & E+E' \ar[l]_{[r,r']} \ar[r]<-.5ex>_{t+t'} \ar[r]
<.5ex>^{s+s'} & N+N' \ar[r] <.5ex>^{[q,q']} & (0,\infty) }.\] 
With this additional map $G$ is a lax monoidal functor.
\end{lem}
\begin{proof}
  We must check the functoriality of $G$, the naturality of the coherence maps, 
  and the coherence axioms. This is similar to the arguments for the
  functor $F$; we omit the details.
\end{proof}

This allows us to make the following definition.

\begin{defn}
The category $\DetBalMark$ is the decorated cospan category where an object is a
finite set with populations and a morphism is an isomorphism class of open
detailed balanced Markov processes $M \maps X \to Y$.  
\end{defn}

As is the case for all decorated cospan categories, we have:

\begin{prop} \label{DetBalMark_dagger}
The category $\DetBalMark$ is a dagger compact category.
\end{prop}

Forgetting the detailed balanced part of a detailed balanced Markov process is
functorial:

\begin{prop}
  There is a `forgetful' faithful symmetric monoidal dagger functor $\DetBalMark \to
  \Mark$ mapping each finite set with populations to its underlying set, and
  each open detailed balanced Markov process to its underlying open Markov
  process.
\end{prop}
\begin{proof}
  This is a straightforward application of Lemma \ref{lemma:decoratedfunctors}.
  Indeed, it is easy to check that we may define a monoidal natural transformation
  \[ 
    \xymatrix{ 
      (\FinPopSet, + ) \ar[dd]_{U} \ar[drr]^{G}
      \ddrtwocell<\omit>{<0>_{\theta}} && \\
      && (\Set, \times) \\
      (\FinSet, + ) \ar[urr]_{F} &&
    } 
  \]
  mapping a detailed balanced Markov process   
\[ \xymatrix{ (0,\infty) & E \ar[l]_-r \ar[r]<-.5ex>_t \ar[r] <.5ex>^s & N \ar[r] <.5ex>^q & (0,\infty)  } 
\]
  in $G(N,q)$ to the Markov process
\[ \xymatrix{ (0,\infty) & E \ar[l]_-r \ar[r]<-.5ex>_t \ar[r] <.5ex>^s & N} 
\]
  in $F(U(N,q)) = F(N)$. Lemma \ref{lemma:decoratedfunctors} then allows us to
  build the desired symmetric monoidal dagger functor. By inspection the
  functor is faithful.
\end{proof}

\section{Black boxing}
\label{sec:black_boxing}

In this section we explain a `black box functor' that sends any open detailed balanced Markov process to a description of its steady-state behavior.

The first point of order is to define the category in which this steady-state
behavior lives. This is the category of linear relations. Already we have seen
that the steady states for an open detailed balanced Markov
process---that is, the solutions to the open master equation---form a linear
subspace of the vector space $\R^X \oplus \R^X \oplus \R^Y \oplus \R^Y$ of input 
and output populations and flows. To compose the steady states of two open Markov
processes is easy: we simply require that the populations and flows at the
outputs of our first Markov process are equal to the populations and flows at
the corresponding inputs of the second. It is also intuitive: it simply means
that we require any states we identify to have identical populations, and
require that at each output state all the outflow from the first Markov process 
flows into the second Markov process.

Luckily, this notion of composition for linear subspaces is already
well known: it is composition of linear relations.  We thus define the following
category:
\begin{defn}
  The category of linear relations, $\LinRel$, has finite-dimensional real
  vector spaces as objects and \define{linear relations} $L \maps U \leadsto V$,
  that is, linear subspaces $L \subseteq U \oplus V$, as morphisms from $U$ to $V$.
  The composite of linear relations $L \maps U \leadsto V$ and $L' \maps V \leadsto
  W$ is given by
  \[   L'L = \{ (u,w) : \exists v \in V \;\, (u,v) \in L \textrm{ and } (v,w) \in L' \}. \]
\end{defn}

\begin{prop}
  The category $\LinRel$ is a dagger compact category.
\end{prop}
\begin{proof}
  This is well known \cite{BE}. The tensor product is given by direct sum: if $L \maps U 
  \leadsto V$ and $L' \maps U' \leadsto V'$, then $L \oplus L' \maps U \oplus U' \leadsto
  V \oplus V'$ is the direct sum of the subspaces $L$ and $L'$.   The dagger is given 
  by relational transpose: if $L \maps U \leadsto V$, then 
  \[   L^\dagger = \{(v,u) : \; (u,v) \in L \} .  \qedhere \]
\end{proof}

\begin{defn}
The \define{black box functor} for detailed balanced Markov processes
\[   \square \maps \DetBalMark \to \LinRel  \]
maps each finite set with populations $(N,q)$ to the vector space
 \[\square(N,q)=  \R^N \oplus \R^N\]
of boundary populations and boundary flows, and each open detailed balanced
Markov process $M \maps X \to Y$ to its behavior
\[ 
  \square(M) = S[i,o](\mathrm{Graph}(\nabla D)) \maps \R^X \oplus
  \R^X \leadsto \R^Y \oplus \R^Y,
\]
where $D$ is the dissipation functional of $M$. 
\end{defn}

We still need to prove that this construction actually gives a functor.  We do this
in Theorem \ref{thm:commuting_triangle} by relating this construction to 
the black box functor for circuits
\[  \blacksquare \maps \Circ \to \LinRel  ,\]
which we studied in a previous paper \cite{BaezFong}.

To define the functor $\blacksquare$, we first construct a decorated cospan category
$\Circ$ in which the morphisms are open circuits.  In brief, let 
\[
  H\maps (\FinSet,+) \longrightarrow (\Set,\times)
\]
map each finite set $N$ to the set of circuits 
\[ \xymatrix{  (0,\infty) & E \ar[l]_-c \ar[r]<-.5ex>_t  \ar[r] <.5ex>^s & N }  \]
on $N$. This can be equipped with coherence maps to form a lax monoidal functor
in the same manner as Markov processes. Using this lax monoidal functor $H$, we
make the following definition.

\begin{defn}
The category $\Circ$ is the decorated cospan category where an object is a
finite set and a morphism is an isomorphism class of open circuits $C \maps X
\to Y$. 
\end{defn}

Again, we will often refer to a morphism as simply an open circuit;
we mean as usual the isomorphism class of the open circuit.

\begin{cor}
The category $\Circ$ is a dagger compact category.
\end{cor}

The main result of our previous paper was this:

\begin{lem} \label{lem:blackbox}
  There exists a symmetric monoidal dagger functor, the \define{black box functor}  
for circuits: 
  \[ \blacksquare\maps \Circ \to \LinRel, \]
   mapping any finite set $X$ to the vector space
\[  \blacksquare(X) = \R^X \oplus \R^X, \] 
and any open circuit $C \maps X \to Y$ to its \define{behavior}, the linear relation 
  \[
   \blacksquare(C) = S[i,o](\mathrm{Graph}(\nabla Q))\maps \R^X \oplus \R^X
   \leadsto \R^Y \oplus \R^Y
  \]
where $Q$ is the power functional of $C$.
\end{lem}
\begin{proof}
This is a simplified version of \cite[Theorem 1.1]{BaezFong}.  Note that in 
Definition \ref{defn:behavior} we defined the behavior of $C$ to be the subspace
\[    S[i,o](\mathrm{Graph}(\nabla Q))  \subseteq 
\R^X \oplus \R^X \oplus \R^Y \oplus \R^Y .  \]
Now we are treating this subspace as a linear relation from $\blacksquare(X)$
to $\blacksquare(Y)$.
\end{proof}

\section{The functor from detailed balanced Markov processes to circuits}
\label{sec:reduction_2}

In Section \ref{sec:reduction} we described a way to model an open detailed
balanced Markov process using an open circuit, motivated by similarities between
dissipation and power. We now show that the analogy between these two structures
runs even deeper: first, this modelling process is functorial, and second, the
behaviors of corresponding Markov processes and circuits are naturally
isomorphic.

\begin{lem} \label{lem:K}
There is a symmetric monoidal dagger functor
\[ K \maps \DetBalMark \to \Circ 
\]
which maps a finite set with populations $(N,q)$ to the underlying finite set
$N$, and an open detailed balanced Markov process 
\[ 
  \xymatrix{ 
  && X \ar[d]^{i} \\
(0,\infty) & E \ar[l]_-r \ar[r]<-.5ex>_t \ar[r] <.5ex>^s & N \ar[r] <.5ex>^q &
(0,\infty) \\ 
&& Y \ar[u]_{o}} 
\]
to the open circuit
\[ 
  \xymatrix{  
    && X \ar[d]^{i} \\
    (0,\infty) & E \ar[l]_-{c} \ar[r]<-.5ex>_t  \ar[r] <.5ex>^s & N
    &\mbox{\phantom{$(0,\infty)$}} \\
    && Y. \ar[u]_{o}
  }
\]
where 
\[  c_e = \frac{1}{2} r_e q_{s(e)}. \]
\end{lem}

\begin{proof}  
  This is another simple application of Lemma \ref{lemma:decoratedfunctors}.
  To see that this gives a functor between the decorated cospan categories we need
  only check that the above function from detailed balanced Markov processes to
  circuits defines a monoidal natural transformation
  \[ 
    \xymatrix{ (\FinPopSet, + ) \ar[dd]_{(U,\upsilon)} \ar[drr]^{(G,\phi)} \ddrtwocell<\omit>{<0>_{\theta}} && \\
    && (\Set, \times) \\
    (\FinSet, + ) \ar[urr]_{(H,\phi')} && 
    } 
  \]
This is easy to check.
\end{proof}

In the above we have described two maps sending an open detailed balanced Markov process to a linear relation:
\[ \blacksquare \circ K \maps \DetBalMark \to \LinRel \]
and 
\[ \square \maps \DetBalMark \to \LinRel .\]
We know the first is a functor; for this second this remains to be proved.  We do this in the process of proving that these two maps are naturally isomorphic:

\begin{thm}
\label{thm:commuting_triangle}
There is a triangle of symmetric monoidal dagger functors \break between dagger compact categories: 
\[
   \xy
   (-20,20)*+{\DetBalMark}="1";
  (20,20)*+{\Circ}="2";
   (0,-20)*+{\LinRel}="5";
        {\ar^{K} "1";"2"};
        {\ar_{\square} "1";"5"};
        {\ar^{\blacksquare} "2";"5"};
        {\ar@{=>}^<<{\scriptstyle \alpha} (1,11); (-3,8)};
\endxy
\]
which commutes up to a monoidal natural isomorphism $\alpha$.  This natural
isomorphism assigns to each finite set with populations $(X,q)$ 
the linear relation $\alpha_{X,q}$ given by the linear map
\begin{align*}
  \alpha_{X,q}\maps \R^X \oplus \R^X &\longrightarrow \R^X \oplus \R^X \\
(\phi,\iota) &\longmapsto (q\phi,\iota)
\end{align*}
where $q\phi \in \R^X$ is the pointwise product of $q$ and $\phi$.
\end{thm}
\begin{proof}
We begin by simultaneously proving the functoriality of $\square$ and the naturality of
$\alpha$. The key observation is that we have the equality
\[   \square(M) = \alpha_{Y,r} \circ \blacksquare K(M) \circ \alpha_{X,q}^{-1}  \]
of linear relations $\R^X \oplus \R^X \to \R^Y \oplus \R^Y$. This is an
immediate consequence of Theorem \ref{thm.behaviors}, which relates the behavior
of the Markov process and the circuit:
\begin{align*}
  \square(M) &= S[i,o]\big(\Graph(\nabla D_M)\big) \\ 
  &= \big(\alpha_{X,q}\oplus \alpha_{Y,r}\big)S[i,o]\big(\Graph(\nabla Q_{K(M)})\big) \\
  &= \alpha_{Y,r} \circ S[i,o]\big(\Graph(\nabla Q_{K(M)})\big) \circ
  \alpha_{X,q}^{-1} \tag{$\ast$}\\
  &= \alpha_{Y,r} \circ \blacksquare K(M) \circ \alpha_{X,q}^{-1}
\end{align*}
The equation ($\ast$) may look a little unfamiliar, but is simply a switch
between two points of view: in the line above we apply the functions $\alpha$ to
the behavior, in the line below we compose the relations $\alpha$ with the
behavior. In either case the same subspace is obtained.

Another way of stating this `key observation' is as the commutativity of the
naturality square
\[
  \xymatrix{
    \R^X \oplus \R^X \ar[rr]^{\blacksquare K(M)} \ar[dd]_{\alpha_{X,q}} &&  \R^Y \oplus \R^Y 
    \ar[dd]^{\alpha_{Y,r}} \\ \\ 
    \R^X \oplus \R^X \ar[rr]^{\square(M)} &&  \R^Y \oplus \R^Y 
  }
\]
for $\alpha$. Thus if $\square$ is truly a functor, then $\alpha$ is a natural
transformation.

But the functoriality of $\square$ is now a consequence of the functoriality of
$\blacksquare$ and $K$. Indeed, for $M \maps (X,q) \to (Y,r)$ and $M' \maps (Y,r) \to (Z,s)$, we
have
\begin{align*}
 \square(M')\circ \square(M) 
 &=  \alpha_{Z,s} \circ \blacksquare K(M') \circ \alpha_{Y,r}^{-1} \circ \alpha_{Y,r} \circ
 \blacksquare K(M) \circ \alpha_{X,q}^{-1} \\
 &= \alpha_{Z,s} \circ \blacksquare K(M') \circ \blacksquare K(M) \circ
 \alpha_{X,q}^{-1} \\
 &= \alpha_{Z,s} \circ \blacksquare K(M' \circ M) \circ
 \alpha_{X,q}^{-1} \\
 &= \square(M'\circ M).
 \end{align*}
Thus $\alpha$ is a natural transformation. It is easily seen that $\alpha$ is
furthermore monoidal, and an isomorphism.  

As a consequence, the functor $\square$ can be given the structure of a symmetric 
monoidal dagger functor, in a way that makes the triangle commute up to $\alpha$.
\end{proof}

\section{Geometrical aspects}
\label{sec:symplectic}

For a physical system whose behavior is described by a variational principle, the
relation between inputs and outputs is typically a Lagrangian relation between symplectic manifolds \cite{Weinstein}.  For example, in a classical system of particles, the positions and momenta of all the particles determine a point in a symplectic manifold.  Thanks to the principle of least action, the relation between position--momentum pairs at one time and another time is a Lagrangian relation.  In our previous work we have seen that because circuits obey the principle of minimum power, black boxing such a circuit gives a Lagrangian relation between potential--current pairs \cite{BaezFong}.  Since detailed balanced Markov processes obey the principle of minimum dissipation, we expect an analogous result for these.  This is what we prove now.

For a review of linear Lagrangian relations, see our previous paper \cite[Section 6]{BaezFong}.  We recall the ideas briefly here.  For us a \define{symplectic vector space} will be a finite-dimensional real vector space $V$ equipped with a nondegenerate antisymmetric bilinear form $\omega \maps V \times V \to \R$.   Given a symplectic vector space $(V,\omega)$, a linear subspace $L \subseteq V$ is \define{Lagrangian} if $L$ is a maximal linear subspace on which $\omega$ vanishes.   Any symplectic vector space $(V,\omega)$ has a \define{conjugate}, the same vector space equipped with the form $\overline{\omega} = - \omega$.    We usually denote a symplectic vector space by a single letter such as $V$,  and denote its conjugate by $\overline{V}$.  Given two symplectic vector spaces their direct sum is naturally a symplectic vector space.  Using all this, we define a \define{Lagrangian relation} $L \maps V \leadsto W$ between symplectic vector spaces to be a linear relation such that $L$ is a Lagrangian subspace of $\overline{V} \oplus W$.  Lagrangian relations are closed under composition.  The category with symplectic vector spaces as objects and Lagrangian relations as morphisms is a dagger compact category, which we call $\LagRel$.

Recall that an object of $\Circ$ is just a finite set $X$, and the black box functor maps this to the space of potential--current pairs:
\[   \blacksquare(X) = \R^X \oplus \R^X . \]
We make this into a symplectic vector space as follows:
\[    \omega\big((\phi, \iota), (\phi', \iota')\big) = \langle \iota', \phi \rangle - \langle \iota , \phi' \rangle ,\]
where the angle brackets denote the standard inner product on $\R^X$:
\[   \langle v, w \rangle = \sum_{i \in X} v_i w_i .\]

We then have:

\begin{lem} \label{lem:symplectic1}
 The black box functor $ \blacksquare\maps \Circ \to \LinRel$ maps any open
circuit $C \maps X \to Y$ to a Lagrangian relation
\[    \blacksquare(C) \maps \blacksquare(X) \leadsto \blacksquare(Y) .
\]
\end{lem}

\begin{proof} This is \cite[Theorem 1.1]{BaezFong}, expressed in slightly different notation: there we write $\blacksquare(X)$ as $\mathbb{R}^X \oplus (\mathbb{R}^X)^*$, which allows us to avoid using an inner product to define the symplectic structure on $\blacksquare(X)$.
\end{proof}

This result implies an analogous result for detailed balanced Markov processes.  An object of $\DetBalMark$ is a finite set with populations $(X,q)$.  The black box functor maps this to a space of population-flow pairs:
\[   \square(X,q) = \R^X \oplus \R^X . \]
We make this into a symplectic vector space as follows:
\[    \omega_q\big((p, j), (p', j')\big) = \langle j', p \rangle_q - \langle j , p' \rangle_q \]
where the angle brackets with a subscript $q$ denote a modified inner product on $\R^X$:
\[   \langle v, w \rangle_q = \sum_{i \in X} q_i^{-1} v_i w_i .\]
Note that 
\begin{equation}
    \omega_q\big((p, j), (p', j')\big) = \omega\big((p/q,j), (p'/q,j')\big).
\label{two_symplectic_structures}
\end{equation}
This relation between symplectic structures, together with the relation between
the two black box functors, yields the following result:

\begin{thm} \label{lem:symplectic2}
The black box functor $\square\maps \DetBalMark \to \LinRel$ maps any open
detailed balanced Markov process $M \maps (X,q) \to (Y,v)$ to a Lagrangian relation
\[    \square(M) \maps \square(X,q) \leadsto \square(Y,v) . 
\]
\end{thm}

\begin{proof}
The key is to use Theorem \ref{thm:commuting_triangle}, which says that
\[   \square(M) = \alpha_{Y,r} \circ \blacksquare K(M) \circ \alpha_{X,q}^{-1}. \]
By Lemma \ref{lem:symplectic1} we know that 
\[   \blacksquare K(M) \maps \blacksquare(X) \leadsto \blacksquare(Y)  \]
is a Lagrangian relation.   Thus, to show $\square(M)$ is a Lagrangian relation,
it suffices to show that for any finite set with populations $(X,q)$, the linear
relation
\[   \alpha_{X,q} \maps \blacksquare(X) \leadsto \square(X,q)  \]
is Lagrangian.  Then the same will be true for $\alpha_{Y,r}$, and $\square(M)$
will be a composite of Lagrangian relations, hence Lagrangian itself.

In fact $\alpha_{X,q}$ is an isomorphism of vector spaces.  It is well
known that an isomorphism of symplectic vector spaces, preserving the symplectic 
structure, defines a Lagrangian relation.  So, it suffices to show that $\alpha_{X,q}$ preserves the symplectic structure:
\[     \omega_q(\alpha_{X,q}(\phi,\iota), \alpha_{X,q}(\phi', \iota')\big) = 
\omega\big((\phi, \iota), (\phi', \iota')\big) \]
for any choice of $(\phi, \iota), (\phi', \iota') \in \R^X \oplus \R^X$.  Recall that
\[  \alpha_{X,q}(\phi, \iota) = (q \phi, \iota).  
\]
Thus, it suffices to show that
\[   \omega_q\big((q \phi, \iota), (q \phi', \iota')\big) = \omega\big((\phi, \iota), (\phi, \iota)\big) .
\]
This follows from Equation \ref{two_symplectic_structures}.
\end{proof}

The modified inner product $\langle v, w \rangle_q$ seems to be a fundamental geometric structure associated to a set with populations.  For example, suppose we have a detailed balanced Markov process on a set with populations $(N,q)$.  Then we saw in Theorem \ref{thm:main} that the master
equation can be written as
\[   \frac{dp_n}{dt} = -q_n \frac{\partial C}{\partial p_n}  \]
where $n$ ranges over $N$.  We can rewrite this as a gradient flow equation using
the modified inner product on $\R^N$.  To do this, define the metric tensor
\[    g_{mn} = \langle e_m, e_n \rangle_q   \]
where $e_m \in \R^N$ is the standard basis vector taking the value $1$ at $m \in N$
and zero elsewhere, and similarly for $e_n$.  Then concretely we have
\[   g_{mn} = \left\{ \begin{array}{cl}
q_n^{-1} & m = n \\ 0 & m \ne n. \end{array} \right.  \]
Using this metric on $\R^N$ we can convert any 1-form on $\R^N$ into a vector field.
Thus, we can convert the differential $df$ of any function $f \maps \R^n \to \R$ into a vector field $\nabla f$.   The master equation then becomes simply
\[       \frac{d p}{d t} = - \nabla C(p)   .\]
This says that as time passes, the population $p$ moves `downhill', opposite to the gradient of $C$.  But here we need the gradient of $C$ defined using the
metric $g$, which is different than the `naive' gradient used in most of this paper.  

\section{Conclusions}
\label{sec:conclusions}

We can summarize the last section by saying that in Theorem \ref{thm:commuting_triangle} we may replace $\LinRel$ by $\LagRel$, the category of symplectic vector spaces and linear relations:
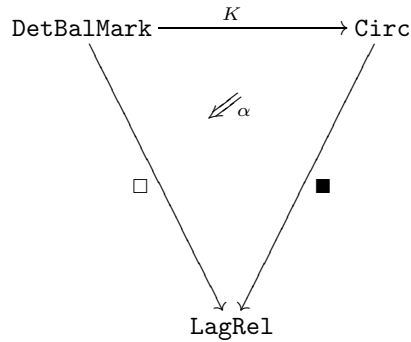
\[
   \xy
   (-20,20)*+{\DetBalMark}="1";
  (20,20)*+{\Circ}="2";
   (0,-20)*+{\LagRel}="5";
        {\ar^{K} "1";"2"};
        {\ar_{\square} "1";"5"};
        {\ar^{\blacksquare} "2";"5"};
        {\ar@{=>}^<<{\scriptstyle \alpha} (1,11); (-3,8)};
\endxy
\]

\noindent
We can also sharpen the analogy chart in the introduction:

\vskip 1em
\begin{center}
\begin{tabular}[h]{|c|c|} \hline
\ \bf Circuits & \bf Detailed balanced Markov processes \\ \hline
Potential: $\phi_i$ & Deviation: $x_i = p_i/q_i$  
\\ \hline
Current: $I_e$  & Flow: $J_e$ 
\\ \hline
Conductance: $c_e$  & Rate constant: $r_e $ 
\\ \hline
Ohm's law: $I_e = c_e (\phi_{s(e)} - \phi_{t(e)})$ &   
Flow law: $J_e = r_e p_{s(e)}$  
\\ \hline
Extended power functional: & Extended dissipation functional: 
\\
 $P(\phi) = \displaystyle{ {\textstyle \frac{1}{2}} \sum_{e_{\;}} c_e \left( \phi_{s(e)} - \phi_{t(e)} \right)^2 } $ & 
  $C(x) = \displaystyle{ {\textstyle \frac{1}{4}} \sum_{e} r_e q_e \left(x_{s(e)} - x_{t(e)}\right)^2 } $ 
\\ \hline
\end{tabular}
\end{center}

\vskip 1em
\noindent
Here we have expressed the dissipation as a function of deviations, rather than
populations.   There are some curious features in this analogy.  They seem to
arise from two facts.  First, in a circuit, the current along an edge depends on
the potential at both the source and target of that edge, while in a Markov
process the flow along an edge depends only on the population at its source.
This can be seen not only in the difference between Ohm's law and the flow law,
but also in the extra factor of one half that appears in the extended
dissipation functional. Indeed, since each edge in a circuit does the job of
two in a Markov process, we must halve the rate when converting it to a
conductance.  Thus, the extended dissipation functional contains a factor of 
$\frac{1}{4}$, while the extended power functional has a factor of $\frac{1}{2}$.

Second, in a circuit, equilibrium is attained when all the potentials $\phi_i$
are equal, while in a detailed balanced Markov process it is not the populations
$p_i$ but the deviations $x_i = p_i/q_i$ that become equal.   Nonetheless the
analogy is close enough that we can make it into a functor, namely $K$, together
with a natural isomorphism relating the black boxing of circuits and that of
detailed balanced Markov processes.

\subsection*{Acknowledgements}

We thank Jason Erbele for giving this paper a careful reading.  BF would like to thank Hertford College, the Centre for Quantum Technologies, the Clarendon Fund and Santander for their support.  BP would like to thank the Centre for Quantum Technologies and the NSF's East Asia and Pacific Summer Institutes program.

\appendix

\section{Decorated cospan categories}
\label{sec:decorated}

This is a brief introduction to decorated cospan categories and functors
between them, reviewing material in \cite{Fong}.

Decorated cospan categories combine information from two monoidal categories.
The first is a category $\CC$ with finite colimits, where the tensor product
is given by the categorical coproduct. From this category we draw the cospans
that we decorate. The second is a monoidal category $(\D,\otimes)$. The objects
in this category represent collections of possible decorations. We construct the
decorated cospan category from a lax monoidal functor between these categories.

\begin{defn} \label{defn.lmf}
  Let $(\CC,\boxtimes)$, $(\D,\otimes)$ be monoidal categories. A \define{lax
  monoidal functor} 
  \[
    (F,\varphi): (\CC,\boxtimes) \to (\D,\otimes)
  \]
  comprises a functor $F: \CC \to D$ and natural transformations 
  \[
    \varphi_{-,-}: F(-)\otimes F(-) \Rightarrow F(-\boxtimes -),
  \]
  \[
    \varphi_1: 1_\D \Rightarrow F1_{\CC}
  \]
  such that three so-called coherence diagrams commute. These diagrams are 
 \[
   \xymatrix{ F(A) \otimes (F(B) \otimes F(C)) \ar[d]_{\mathrm{id} \otimes \varphi_{B,C}} 
                    \ar[rr]^{\sim} &&
                    (F(A) \otimes F(B)) \otimes F(C) \ar[d]^{\varphi_{A,B}
		    \otimes \mathrm{id}}  \\
                    F(A) \otimes F(B \boxtimes C) \ar[d]_{\varphi_{A,B\boxtimes C}} &&
                     F(A \boxtimes B) \otimes F(C) \ar[d]^{\varphi_{A\boxtimes B,C}} \\
                     F(A \boxtimes (B\boxtimes C)) \ar[rr]^{\sim} && 
                     F((A \boxtimes B) \boxtimes C)                    
    }
  \]
where the horizontal arrows come from the associators for $\otimes$ and
$\boxtimes$, and
\[ 
  \xymatrix{
    1_\D \otimes F(A) \ar[r]^{\varphi_1 \otimes \mathrm{id}} \ar[d]_{\sim} &
    F(1_{\CC}) \otimes F(A) \ar[d]^{\varphi_{1_\CC,A}} \\
    F(A) & F(1_\CC \boxtimes A) \ar[l]_{\sim}
  }
\]
and
\[ 
  \xymatrix{
    F(A) \otimes 1_\D \ar[r]^{\mathrm{id} \otimes \varphi_1} \ar[d]_{\sim} &
    F(A) \otimes F(1_\CC) \ar[d]^{\varphi_{A,1_\CC}} \\
    F(A) & F(A \boxtimes 1_\CC) \ar[l]_{\sim}
  }
\]
where the isomorphisms come from the unitors for $\otimes$ and $\boxtimes$.
\end{defn}

The decorated cospan construction is then as follows:

\begin{lem} \label{lemma:fcospans}
  Suppose $\CC$ is a category with finite colimits and let
  \[
    (F,\varphi)\maps (\CC,+) \longrightarrow (\D, \otimes)
  \]
  be a lax monoidal functor, where $+$ stands for the coproduct in $\CC$. 
   We may define a category $F\Cospan$, the category of
   \define{$F$-decorated cospans}, whose objects are those of $\CC$ and whose
   morphisms are equivalence classes of pairs 
  \[
    (X \stackrel{i}\longrightarrow N \stackrel{o}\longleftarrow Y,\;\; 1
    \stackrel{s}\longrightarrow FN)
  \]
  comprising a cospan $X \stackrel{i}\rightarrow N \stackrel{o}\leftarrow Y$ in
  $\CC$ together with a morphism $1 \stackrel{s}\rightarrow F(N)$.  We call $s$ the    
   \define{decoration} of the decorated
  cospan. The equivalence relation arises from isomorphism of cospans; an 
  isomorphism of cospans induces a one-to-one correspondence between their   decorations.

  Composition in this category is given by pushout of cospans in $\CC$: 
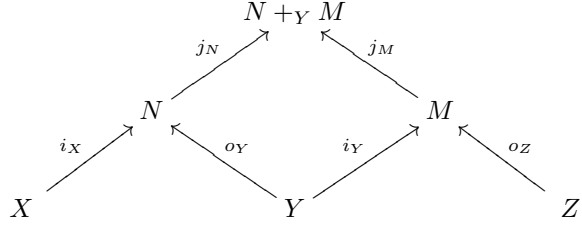
  \[
    \xymatrix{
      && N+_YM \\
      & N \ar[ur]^{j_N} && M \ar[ul]_{j_M} \\
      \quad X\quad \ar[ur]^{i_X} && Y \ar[ul]_{o_Y} \ar[ur]^{i_Y} &&\quad Z \quad \ar[ul]_{o_Z}
    }
  \]
  paired with the pushforward 
  \[
    1 \stackrel{\lambda^{-1}}\longrightarrow 1 \otimes 1 \stackrel{s \otimes
    t}\longrightarrow FN \otimes FM \stackrel{\varphi_{N,M}}\longrightarrow
    F(N+M) \stackrel{F[j_N,j_M]}\longrightarrow F(N+_YM)
  \]
  of the tensor product of the decorations along the coproduct of the pushout maps.
\end{lem}

\begin{proof}
This is \cite[Proposition 3.2]{Fong}.
\end{proof}

When the functor $F$ is lax symmetric monoidal, the category $F\Cospan$
becomes a dagger compact category.  To be more precise, we define the tensor product of objects $X$ and $Y$ of $F\Cospan$ to be their coproduct $X+Y$, and define the tensor product of decorated cospans $(X \stackrel{i_X}\longrightarrow N
\stackrel{o_Y}\longleftarrow Y,\;\; 1 \stackrel{s}\longrightarrow FN)$ and
$(X' \stackrel{i_{X'}}\longrightarrow N' \stackrel{o_{Y'}}\longleftarrow
Y',\;\; 1 \stackrel{t}\longrightarrow FN')$ to be 
\[
  \left(
  \begin{aligned}
    \xymatrix{
      & N+N' \\  
      X+X' \ar[ur]^{i_X+i_{X'}} && Y+Y' \ar[ul]_{o_Y+o_{Y'}}
    }
  \end{aligned}
  ,
  \qquad
  \begin{aligned}
    \xymatrix{
      F(N+N') \\
      1 \ar[u]_{\varphi_{N,N'} \circ (s \otimes t) \circ \lambda^{-1}}
    }
  \end{aligned}
  \right).
\]
We also write $+$ for the tensor product in $F\Cospan$. 

The dagger structure for $F\Cospan$ reflects the cospan part of a decorated cospan, 
while keeping the same decoration: 
\[
  \dagger(X \stackrel{i}\longrightarrow N \stackrel{o}\longleftarrow Y,\;\;
  1 \stackrel{s}\longrightarrow FN) = (Y \stackrel{o}\longrightarrow N
  \stackrel{i}\longleftarrow X,\;\; 1 \stackrel{s}\longrightarrow FN).
\]
This gives the following fact:

\begin{lem} \label{dagger_lemma}
  Let $F$ be a lax symmetric monoidal functor. Then with the above structure, the
  category $F\Cospan$ is a dagger compact category.
\end{lem}
\begin{proof}
This is a special case of \cite[Theorem 3.4]{Fong}.
\end{proof}

Decorated cospans allow us to understand the diagrammatic nature of structures
on finite sets, such as Markov processes.  We also need the ability to map one kind of
structure to another. This is provided by monoidal natural transformations
between our lax monoidal functors.

\begin{defn}
  A \define{monoidal natural transformation} $\alpha$ from a lax monoidal
  functor
  \[    (F,\varphi)\maps (\CC,\boxtimes) \longrightarrow (\D,\otimes)
  \]
  to a lax monoidal functor
  \[
    (G,\gamma)\maps (\CC,\boxtimes) \longrightarrow (\D,\otimes)
  \]
  is a natural transformation $\alpha\maps F \Rightarrow G$ such that
  \[
    \xymatrix{
      F(A) \boxtimes F(B) \ar[r]^{\varphi_{A,B}} \ar[d]_{\alpha_A \boxtimes
      \alpha_B} & F(A \otimes B) \ar[d]^{\alpha_{A\otimes B}} \\
      G(A) \boxtimes G(B) \ar[r]^{\gamma_{A,B}} & G(A \otimes B)
    }
  \]
  commutes.
\end{defn}

We then construct functors between decorated cospan categories as follows:

\begin{lem} \label{lemma:decoratedfunctors}
  Let $\CC$ and $\CC'$ be categories with finite colimits, abusing
  notation to write the coproduct in both categories as $+$, and let $(\D,
  \otimes)$ and $(\D',\boxtimes)$ be symmetric monoidal categories. Further let
  \[
    (F,\varphi) \maps (\CC,+) \longrightarrow (\D,\otimes)
  \]
  and
  \[
    (G,\gamma) \maps (\CC',+) \longrightarrow (\D',\boxtimes)
  \]
  be lax symmetric monoidal functors. This gives rise to decorated cospan
  categories $F\Cospan$ and $G\Cospan$. 

  Suppose then that we have a finite colimit-preserving functor $A \maps \CC \to
  \CC'$ with accommopanying natural isomorphism $\alpha \maps A(-)+A(-)
  \Rightarrow A(-+-)$, a lax monoidal functor $(B,\beta) \maps (\D, \otimes) \to
  (\D', \boxtimes)$, and a monoidal natural transformation $\theta \maps (B
  \circ F, B\varphi\circ\beta) \Rightarrow (G \circ A, G\alpha\circ\gamma)$.
  This may be depicted by the diagram:
  \[
    \xymatrixcolsep{3pc}
    \xymatrixrowsep{3pc}
    \xymatrix{
      (\CC,+) \ar^{(F,\varphi)}[r] \ar_{(A,\alpha)}[d] \drtwocell
      \omit{_\:\theta} & (\D,\otimes) \ar^{(B,\beta)}[d]  \\
      (\CC',+) \ar_{(G,\gamma)}[r] & (\D',\boxtimes).
    }
  \]

  Then we may construct a symmetric monoidal dagger functor 
  \[
    (T, \tau): F\Cospan \longrightarrow G\Cospan
  \]
  mapping each object $X \in F\Cospan$ to $AX \in G\Cospan$, and
  each morphism
  \[
    (X \stackrel{i}\longrightarrow N \stackrel{o}\longleftarrow Y,\quad 1_{\D}
    \stackrel{s}\longrightarrow FN)
  \]
  to
  \[
    (AX \stackrel{Ai}\longrightarrow AN \stackrel{Ao}\longleftarrow AY,\quad 1_{\D'}
    \stackrel{\theta_N\circ Bs\circ\beta_1}\longrightarrow GAN).
  \]
\end{lem}
\begin{proof}
This is a special case of \cite[Theorem 4.1]{Fong}.
\end{proof}

\end{document}